\documentclass[a4paper,UKenglish]{lipics-v2021}

\usepackage{macros}

\title{Parameterised Counting in Logspace} 

\author{Anselm Haak}{Leibniz Universität Hannover, Institut für Theoretische Informatik, Hannover, Germany}{haak@thi.uni-hannover.de}{https://orcid.org/0000-0003-1031-5922}{}

\author{Arne Meier}{Leibniz Universität Hannover, Institut für Theoretische Informatik, Hannover, Germany}{meier@thi.uni-hannover.de}{https://orcid.org/0000-0002-8061-5376}{Funded by the German Research Foundation (DFG), project ME4279/1-2}

\author{Om Prakash}{Department of Computer Science and Engineering, IIT Madras, Chennai, India}{op708543@gmail.com}{}{}

\author{Raghavendra Rao B.\ V.}{Department of Computer Science and Engineering, IIT Madras, Chennai, India}{bvrr@cse.iitm.ac.in}{}{}

\authorrunning{A. Haak, A. Meier, O. Prakash, and Raghavendra Rao B.\ V.}

\Copyright{Anselm Haak, Arne Meier, Om Prakash, and Raghavendra Rao B.\ V.}



\nolinenumbers
\funding{Indo-German co-operation grant: DAAD (57388253), DST (INT/FRG/DAAD/P-19/2018)}

\acknowledgements{The authors thank the anonymous referees for their valuable feedback.}

\hideLIPIcs

\EventEditors{Markus Bl\"{a}ser and Benjamin Monmege}
\EventNoEds{2}
\EventLongTitle{38th International Symposium on Theoretical Aspects of Computer Science (STACS 2021)}
\EventShortTitle{STACS 2021}
\EventAcronym{STACS}
\EventYear{2021}
\EventDate{March 16--19, 2021}
\EventLocation{Saarbr\"{u}cken, Germany (Virtual Conference)}
\EventLogo{}
\SeriesVolume{187}
\ArticleNo{45}

\begin{document}

\maketitle

\begin{abstract}
Logarithmic space bounded complexity classes such as $\L$ and $\NL$ play a central role in space bounded computation. 
The study of counting versions of these complexity classes have lead to several interesting insights into the structure of computational problems such as computing the determinant and counting paths in directed acyclic graphs. 
Though parameterised complexity theory was initiated roughly three decades ago by Downey and Fellows, a satisfactory study of parameterised logarithmic space bounded computation was developed only in the last decade by Elberfeld, Stockhusen and Tantau (IPEC 2013, Algorithmica 2015). 

In this paper, we introduce a new framework for parameterised counting in logspace, inspired by the parameterised space bounded models developed by Elberfeld, Stockhusen and Tantau (IPEC 2013, Algorithmica 2015). 
They defined the operators $\paraw$ and $\parab$ for parameterised space complexity classes by allowing bounded nondeterminism with multiple-read and read-once access, respectively. 
Using these operators, they characterised the parameterised complexity of natural problems on graphs. 
In the spirit of the operators $\paraw$ and $\parab$ by Stockhusen and Tantau, we introduce variants based on tail-nondeterminism, $\parawt$ and $\parabt$.
Then, we consider counting versions of all four operators applied to logspace and obtain several natural complete problems for the resulting classes: counting of paths in digraphs, counting first-order models for formulas, and counting graph homomorphisms. 
Furthermore, we show that the complexity of a parameterised variant of the determinant function for $(0,1)$-matrices is $\sh\parabt\L$-hard and can be written as the difference of two functions in $\sh\parabt\L$. 
These problems exhibit the richness of the introduced counting classes. 
Our results further indicate interesting structural characteristics of these classes. 
For example, we show that the closure of $\sh\parabt\L$ under parameterised logspace parsimonious reductions coincides with $\sh\parab\L$, that is, modulo parameterised reductions, tail-nondeterminism with read-once access is the same as read-once nondeterminism.  

Initiating the study of closure properties of these parameterised logspace counting classes, we show that all introduced classes are closed under addition and multiplication, and those without tail-nondeterminism are closed under parameterised logspace parsimonious reductions. 

Also, we show that the counting classes defined can naturally be characterised by parameterised variants of classes based on branching programs in analogy to the classical counting classes. 

Finally, we underline the significance of this topic by providing a promising outlook showing several open problems and options for further directions of research.
\end{abstract}
\section{Introduction}
\label{sec:intro}
Parameterised complexity theory, introduced by Downey and Fellows~\cite{DBLP:series/mcs/DowneyF99}, takes a two-dimensional view on the computational complexity of problems and has revolutionised the algorithmic world. 
Two-dimensional here refers to the fact that the complexity of a parameterised problem is analysed with respect to the input size $n$ and a parameter $k$ associated with the given input as two independent quantities. 
The notion of fixed-parameter tractability is the proposed notion of efficient computation. 
A problem with parameter $k$ is fixed-parameter tractable (fpt, or in the class $\FPT$) if there is a deterministic $f(k)\cdot n^{O(1)}$ time algorithm for deciding it, where $f$ is a computable function. 
The primary notion of intractability is captured by the $\W{}$-hierarchy in this setting. 

Since its inception, the focus of parameterised complexity theory has been to identify parameterisations of \NP-hard problems that allow for efficient parameterised algorithms, and to address structural aspects of the classes in the $\W{}$-hierarchy and related complexity classes~\cite{DBLP:series/txtcs/FlumG06}.
This led to the development of machine-based and logical characterisations of parameterised complexity classes (see the book by Flum and Grohe~\cite{DBLP:series/txtcs/FlumG06} for more details). 
While the structure of classes in hierarchies such as the $\W{}$- and $\A{}$-hierarchy is well understood, a parameterised view of parallel and space-bounded computation lacked attention. 

In 2013, Elberfeld~et~al.~\cite{DBLP:conf/iwpec/StockhusenT13,DBLP:journals/algorithmica/ElberfeldST15} focused on parameterised space complexity classes and initiated the study of parameterised circuit complexity classes. 
In fact, they introduced parameterised analogues of deterministic and nondeterministic logarithmic space-bounded classes. 
The machine-based characterisation of $\W\P$ (the class of problems that are fpt-reducible to a weighted circuit satisfiability question), and the type of access to nondeterministic choices (multi-read or read-once) led to two different variants of parameterised logspace (para-logspace), namely, $\paraw\L$ and $\parab\L$. 
Elberfeld~et~al.~\cite{DBLP:journals/algorithmica/ElberfeldST15} obtained several natural complete problems for these classes, such as parameterised variants of reachability in graphs.

Bannach, Stockhusen and Tantau~\cite{DBLP:conf/iwpec/BannachST15} further studied parameterised parallel algorithms. 
They used colour coding techniques~\cite{DBLP:journals/jacm/AlonYZ95} to obtain efficient parameterised parallel algorithms for several natural problems. 
A year later, Chen and Flum~\cite{DBLP:conf/mfcs/ChenF16} proved parameterised lower bounds for $\AC0$ by adapting circuit lower bound techniques.

Apart from decision problems, counting problems have found a prominent place in complexity theory. 
Valiant~\cite{DBLP:journals/tcs/Valiant79} introduced the notion of counting complexity classes that capture natural counting problems such as counting the number of perfect matchings in a graph, or counting the number of satisfying assignments of a CNF formula.
Informally, $\sh\P$ (resp., $\sh\L$) consists of all functions $F\colon\{0,1\}^*\to\N$ such that there exists an nondeterministic Turing machine (NTM) running in polynomial time (resp., logarithmic space) in the input length whose number of accepting paths on every input $x\in\{0,1\}^*$ is equal to $F(x)$. 
Valiant's theory of $\sh\P$-completeness led to several structural insights into complexity classes around $\NP$ and interactive proof systems, as well as to the seminal result of Toda~\cite{DBLP:journals/siamcomp/Toda91}. 
 
While counting problems in $\sh\P$ stayed in the focus of research for long, the study of the determinant by Damm~\cite{Dam91}, Vinay~\cite{ctc/Vin91}, and Toda~\cite{Toda91CountingPC} established that the complexity of computing the determinant of an integer matrix characterises the class $\sh\L$ up to a closure under subtraction. 
Allender and Ogihara~\cite{DBLP:journals/ita/AllenderO96} analysed the structure of complexity classes based on $\sh\L$. 
The importance of counting classes based on logspace-bounded Turing machines (TMs) was further established by Allender, Beals and Ogihara~\cite{DBLP:journals/cc/AllenderBO99}.
They characterised the complexity of testing feasibility of linear equations by a class which is based on $\sh\L$. 
Beigel and Fu~\cite{DBLP:journals/jcss/BeigelF00} then showed that small depth circuits built with oracle access to $\sh\L$ functions lead to a hierarchy of languages which can be seen as the logspace version of the counting hierarchy. 
In a remarkable result, Ogihara~\cite{DBLP:journals/siamcomp/Ogihara98} showed that this hierarchy collapses to the first level.
Further down the complexity hierarchy, Caussinus~et~al.~\cite{DBLP:journals/jcss/CaussinusMTV98} introduced counting versions of $\NC1$ based on various characterisations of $\NC1$. 
The counting and probabilistic analogues of $\NC1$ exhibit properties similar to their logspace counterparts~\cite{DBLP:journals/tcs/DattaMRTV12}. 
Moreover, counting and gap variants of the class $\AC0$ were defined by Agrawal~et~al.~\cite{DBLP:journals/jcss/AgrawalAD00}. 
 
The theory of parameterised counting classes was pioneered by Flum and Grohe~\cite{DBLP:journals/siamcomp/FlumG04} as well as McCartin~\cite{DBLP:conf/mfcs/McCartin02}. The class $\sh\W{1}$ consists of all parameterised counting problems that reduce to the problem of counting $k$-cliques in a graph.  
Flum and Grohe~\cite{DBLP:journals/siamcomp/FlumG04} proved that counting cycles of length $k$ is complete for $\sh\W{1}$. 
Curticapean~\cite{DBLP:conf/icalp/Curticapean13} further showed that counting matchings with $k$ edges in a graph is also complete for $\sh\W{1}$. 
These results led to several remarkable completeness results and new techniques (see, e.g., the works of Curticapean~\cite{DBLP:phd/dnb/Curticapean15,DBLP:journals/iandc/Curticapean18}, Curticapean, Dell and Marx~\cite{DBLP:conf/stoc/CurticapeanDM17}, Jerrum and Meeks~\cite{DBLP:journals/combinatorica/JerrumM17}, Brand and Roth~\cite{DBLP:conf/csr/BrandR17}).
 
\subparagraph*{Motivation} 
Given the rich structure of logspace-bounded counting complexity classes, the study of parameterised variants of these classes is vital to obtain a finer classification of counting problems.  

A theory on para-logspace counting did not exist before. 
We wanted to overcome this defect to further understand the landscape of counting problems with decision versions in para-logspace-based classes. 
Our new framework allows us to classify many of these problems more precisely.
In this article, we define counting variants inspired by the parameterised space complexity classes introduced by Elberfeld et al.~\cite{DBLP:conf/iwpec/StockhusenT13,DBLP:journals/algorithmica/ElberfeldST15}. 

In the realm of space-bounded computation, different manners in which nondeterministic bits are accessed lead to different complexity classes. 
For example, the standard definition of $\NL$ implicitly gives the corresponding NTMs only read-once access to their nondeterministic bits~\cite{DBLP:books/daglib/0023084}: nondeterminism is given only in the form of choices between different transitions.
This means that nondeterministic bits are not re-accessible by the machine later in the computation.
When instead using an auxiliary read-only tape for these bits and allowing for multiple passes on it, one obtains the class $\NP$.
This is due to the fact that $\textbf{SAT}$ is $\NP$-complete with respect to logspace many-one reductions~\cite{DBLP:books/daglib/0023084}, and that one can evaluate a {CNF} formula in deterministic logspace even when the assignment is given on a read-only tape. 
However, polynomial time bounded NTMs still characterise $\NP$ even when the machine is allowed to do only one pass on the nondeterministic bits as they can simply store all nondeterministic bits on the work-tape.
So, it is very natural to investigate whether the differentiation from above leads to new insights in our setting.

With parameterisation as a means for a finer classification, Stockhusen and Tantau~\cite{DBLP:conf/iwpec/StockhusenT13} defined nondeterministic logarithmic space-bounded computation based on \emph{how} (unrestricted or read-once) the nondeterministic bits are accessed. 
Based on this distinction, they defined two operators: $\paraw$ (unrestricted) and $\parab$ (read-once). 
Their study led to many compelling natural problems that are complete for logspace-bounded nondeterministic computations with suitable parameters. 
Thereby, a rich structure of computational power based on the restrictions on the number of reads of the nondeterministic bits was exhibited.
In this article, we additionally differentiate based on \emph{when} (unrestricted or tail access) the nondeterministic bits are accessed.
The classes $\W1$ and $\W\P$ are the two most prominent nondeterministic classes in the parameterised world which is why we wanted to see the effect of such a restriction on the rather small classes in our setting.
This leads to the new operators $\parawt$ and $\parabt$.
The concept of tail-nondeterminism allowed to capture the parameterised complexity class $\W1$---via tail-nondeterministic, $k$-bounded machines---and thereby relates to many interesting problems such as searching for cliques, independent sets, or homomorphism, and evaluating conjunctive queries~\cite{DBLP:series/txtcs/FlumG06}.
Intuitively, tail-nondeterminism means that all nondeterministic bits are read at the end of the computation, and $k$-boundedness limits the number of these nondeterministic bits to $f(k)\cdot\log|x|$ for all inputs $(x,k)$.

Studying counting complexity often improves the understanding of related classical problems and classes (e.g., Toda's theorem~\cite{DBLP:journals/siamcomp/Toda91}).
With regard to space-bounded complexity, there are several characterisations of logspace-bounded counting classes in terms of natural problems. 
For example, counting paths in directed graphs is complete for $\sh\L$, and checking if an integer matrix is singular or not is complete for the class $\Ceq\L$ (see Allender~et~al.~\cite{DBLP:journals/cc/AllenderBO99}). 
Furthermore, testing if a system of linear equations is feasible or not can be done in $\L$ with queries to any complete language for $\Ceq\L$. 
Moreover, two hierarchies built over counting classes for logarithmic space collapse either to the first level~\cite{DBLP:journals/siamcomp/Ogihara98} or to the second level~\cite{DBLP:journals/cc/AllenderBO99}. 
Apart from this, the separation of various counting classes over logarithmic space remains widely open. 
For example, it is not known if the class $\Ceq\L$ is closed under complementation. 

We consider different parameterised variants of the logspace-bounded counting class $\sh\L$ to give a new perspective on its fine structure.

\subparagraph*{Results} 
We introduce the counting variants of parameterised space-bounded computation and show that each of the parameterised logspace complexity classes, defined by Stockhusen and Tantau~\cite{DBLP:conf/iwpec/StockhusenT13}, has a natural counting counterpart.  
Moreover, by considering also tail-nondeterminism with respect to their classes, we obtain four different variants of parameterised logspace counting classes, namely, $\sh\paraw\L, \sh\parab\L, \sh\parawt\L$, and $\sh\parabt\L$. 
We show that $\sh\paraw\L$ and $\sh\parab\L$ are closed under para-logspace parsimonious reductions and that all of our new classes are closed under addition and multiplication.

Furthermore, we develop a complexity theory by obtaining natural complete problems for these new classes.
We introduce variants of the problem of counting walks of parameter-bounded length that are complete for the classes $\sh\parab\L$ (Theorems~\ref{thm:logreach-parabetaL-paralog-complete}, \ref{thm:reach-parabetal-complete} and \ref{thm:pathb-parabetal-complete-turingplog}), $\sh\parabt\L$ (Theorem~\ref{thm:reach-parabetataill-complete}) and $\sh\paraw\L$ (Theorem~\ref{thm:reach2cnf-parawl-complete}).
Since the same problem is shown to be complete for both, $\sh\parab\L$ and $\sh\parabt\L$, we get the somewhat surprising result that the closure of $\sh\parabt\L$ under para-logspace parsimonious reductions coincides with $\sh\parab\L$ (Corollary~\ref{cor:parabtaillog-clos=parabl}). 
Also, we show that a parameterised version of the problem of counting homomorphisms from coloured path structures to arbitrary structures is complete for $\sh\parab\L$ with respect to para-logspace parsimonious reductions (Theorem~\ref{thm:phom(p^*)-pbetaL-complete}).

Afterwards, we study variants of the problem of counting assignments to free first-order variables in a quantifier-free $\FO$ formula.
We identify complete problems for the classes $\sh\parab\L$ and $\sh\parawt\L$ in this context.
More specifically, counting assignments to free first-order variables in a quantifier-free formula with relation symbols of bounded arity and the syntactical locality of the variables in the formula being restricted ($p\text-\sh\MC{\Sigmarlocal}_a$) is shown to be complete for the classes $\sh\parabt\L$ and $\sh\parab\L$ with respect to para-logspace parsimonious reductions (Theorem~\ref{thm:mc-local-bounded}). 
When there is no restriction on the arity of relational symbols or on the locality of the variables, counting the number of satisfying assignments to free first-order variables in a quantifier-free formula in a given structure ($p\text-\sh\MC{\Sigma_0}$) is complete for $\sh\parawt\L$ with respect to para-logspace parsimonious reductions (Theorem~\ref{thm:mc-for-w1}).   

Finally, we consider a parameterised variant of the determinant function ($\pdet$) introduced by Chauhan and Rao~\cite{DBLP:conf/caldam/ChauhanR15}. 
By adapting the arguments of Mahajan and Vinay~\cite{DBLP:journals/cjtcs/MahajanV97}, we show that $\pdet$ on $(0,1)$-matrices can be expressed as the difference of two functions in $\sh\parab\L$, and is $\sh\parabt\L$-hard with respect to para-logspace many-one reductions (Theorem~\ref{thm:pdet-ub}).    
    
Figure~\ref{fig:class-diagram} shows a class diagram with complete problems. 
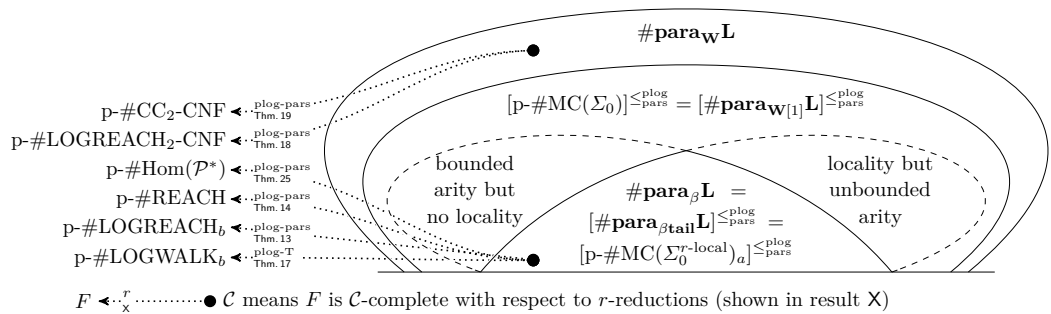
\begin{figure}
	\centering
	\resizebox{\textwidth}{!}{%
	\begin{tikzpicture}[dot/.style={circle,fill=black,inner sep=.5mm}]

\begin{scope}[yshift=-4.5cm,xshift=-1cm]
		\path[black,stealth'-*,thick,dotted] (-9,4) edge [in=180,out=0]  node[pos=0.24,fill=white,inner sep=.1mm,text width=2mm] {$\scriptstyle r$\\[-2mm]\tiny\textsf{X}} (-7,4);
		\node[anchor=west] at (-7,4) {$\calc$ means $F$ is $\calc$-complete \wrt $r$-reductions (shown in result \textsf{X})};
		\node[anchor=east] at (-9,4) {$F$};
\end{scope}
	
		\node[anchor=west,label={180:$\p\sh\cyclecover_2\CNF$}] (cyclecover) at (-7.7,2.75) {};
		\node[anchor=west,label={180:$\p\sh\LOG\reach_2\CNF$}] (reach2cnf) at (-7.7,2.25) {};
		\node[anchor=west,label={180:$\pcountHom{\mathcal P^*}$}] (phom) at (-7.7,1.75) {};
		\node[anchor=west,label={180:$\p\sh\reach$}] (preach)  at (-7.7,1.25) {};
		\node[anchor=west,label={180:$\p\sh\LOG\reach\bdegree$}] (preachlog) at (-7.7,.75) {};
		\node[anchor=west,label={180:$\p\sh\LOG\WALK\bdegree$}] (pwalklog) at (-7.7,.25) {};

		\path[black,stealth'-*,thick,dotted,shorten <=-3mm] (cyclecover) edge [in=180,out=0]  node[pos=0.106,fill=white,inner sep=.1mm,text width=1cm] {\tiny plog-pars\\[-2.5mm]\sffamily Thm.\,\ref{thm:cyclecover-parawl}} (-2.5,3.8);
		
		\path[black,stealth'-*,thick,dotted,shorten <=-3mm] (reach2cnf) edge [in=180,out=0]  node[pos=0.103,fill=white,inner sep=.1mm,text width=1cm] {\tiny plog-pars\\[-2.5mm]\sffamily Thm.\,\ref{thm:reach2cnf-parawl-complete}} (-2.5,3.8);

		\path[black,stealth'-*,thick,dotted,shorten <=-3mm] (phom) edge [in=180,out=0]  node[pos=0.103,fill=white,inner sep=.1mm,text width=1cm] {\tiny plog-pars\\[-2.5mm]\sffamily Thm.\,\ref{thm:phom(p^*)-pbetaL-complete}} (-2.5,.2);		

		\path[black,stealth'-*,thick,dotted,shorten <=-3mm] (preach) edge [in=180,out=0]  node[pos=0.105,fill=white,inner sep=.1mm,text width=1cm] {\tiny plog-pars\\[-2.5mm]\sffamily Thm.\,\ref{thm:reach-parabetal-complete}} (-2.5,.2);
				
		\path[black,stealth'-*,thick,dotted,shorten <=-3mm] (preachlog) edge node[pos=0.12,fill=white,inner sep=.1mm,text width=1cm] {\tiny plog-pars\\[-2.5mm]\sffamily Thm.\,\ref{thm:logreach-parabetaL-paralog-complete}} (-2.5,.2);

		\path[black,stealth'-*,thick,dotted,shorten <=-3mm] (pwalklog) edge [in=180,out=0]  node[pos=0.088,fill=white,inner sep=.1mm,text width=.8cm] {\tiny plog-T\\[-2.5mm]\sffamily Thm.\,\ref{thm:pathb-parabetal-complete-turingplog}} (-2.5,.2);

		\path[use as bounding box] (-6.5,-.1) rectangle (6.5,5);
		
		\node[anchor=south,text width=4.5cm,align=center] at (0,0) {$\#\parab\L=\clospara{\#\parabt\L}=\clospara{\p\#\MC{\Sigmarlocal}_a}$};

		\path[white,ultra thick,shorten >=.4mm,shorten <=1mm] (-3.5,0) edge[out=160,in=130,looseness=1.99] (3.5,0);
		\path[white,ultra thick,shorten <=.4mm,shorten >=1mm] (-3.5,0) edge[out=50,in=20,looseness=1.98] (3.5,0);
		\path[black] (-3.5,0) edge[out=160,in=130,looseness=2] (3.5,0);
		\path[black] (-3.5,0) edge[out=50,in=20,looseness=2] (3.5,0);
		\node[anchor=south,text width=2cm,align=center] at (-3.6,.65) {bounded arity but no locality};
		\node[anchor=south,text width=2cm,align=center] at (3.3,.65) {locality but unbounded arity};

		\path[white,ultra thick] (-4.5,0) edge[out=140,in=40,looseness=2.1] (4.5,0);
		\path[black] (-4.5,0) edge[out=140,in=40,looseness=2.1] (4.5,0);
		\node[anchor=south] at (0,2.5) {$\clospara{\p\#\MC{\Sigma_0}}=\clospara{\#\parawt\L}$};
	
		\path[white,ultra thick] (-4.8,0) edge[out=140,in=40,looseness=2.5] (4.8,0);
		\path[black] (-4.8,0) edge[out=140,in=40,looseness=2.5] (4.8,0);
		\node[anchor=south] at (0,3.8) {$\#\paraw\L$};
		
		\path[white,dashed,ultra thick,shorten <=1mm,shorten >=1mm] (0,2.08) edge[out=13.4,in=20,looseness=2.42] (3.5,0);
		\path[white,dashed,ultra thick,shorten <=1mm,shorten >=1mm] (0,2.08) edge[out=166.6,in=160,looseness=2.42] (-3.5,0);

		\draw[black] (-5.25,0) -- (5.25,0);
	
	\end{tikzpicture}}
	\caption{Diagram assuming pair-wise difference of studied classes with list of complete problems.}\label{fig:class-diagram}
\end{figure}

 \subparagraph*{Main Techniques}
 Our primary contribution is laying foundations for the study of parameterised logspace-bounded counting complexity classes.  
 The completeness results in Theorems~\ref{thm:reach-parabetal-complete} and~\ref{thm:mc-for-w1} required a quantised normal form for $k$-bounded nondeterministic Turing Machines (NTMs) (Lemma~\ref{normalise-lemma}). 
 This normal form quantises the nondeterministic steps of a $k$-bounded NTM into chunks of $\log n$-many steps such that the total number of accepting paths remains the same. 
 We believe that the normal form given in Lemma~\ref{normalise-lemma} will be useful in the structural study of parameterised counting classes.   
 The study of $\pdet$ involved definitions of so-called parameterised clow sequences generalising the classical notion \cite{DBLP:journals/cjtcs/MahajanV97}. 
 Besides, a careful assignment of signs to clow sequences was necessary for our complexity analysis of $\pdet$.

\subparagraph*{Related Results}
Chen and M\"uller~\cite{DBLP:journals/toct/ChenM15} studied the parameterised complexity of counting homomorphisms and divided the problems into four equivalence classes. 
However, their equivalence is only based on reductions among variants of counting homomorphisms but not in terms of concrete complexity classes.  
In this context, Dalmau and Johnson~\cite{DBLP:journals/tcs/DalmauJ04} investigated the complexity of counting homomorphisms as well, and provided generalisations of results by Grohe~\cite{DBLP:journals/jacm/Grohe07} to the counting setting.
A similar classification regarding our classes can give new insights into the complexity of the homomorphism problem (Open Problem~\ref{op:phom-natural-w1L-complete}).
The behaviour of our classes with respect to reductions is similar to the one observed for $\W{1}$ by Bottesch~\cite{DBLP:conf/iwpec/Bottesch17,DBLP:conf/mfcs/Bottesch18}. 

\subparagraph*{Outline}
In Section~\ref{sec:prelims}, we introduce the considered machine model, as well as needed foundations of parameterised complexity theory, and logic.
Section~\ref{sec:class} presents structural results regarding our introduced notions in the parameterised counting context.
Afterwards, in Section~\ref{sec:complete}, our main results on counting walks, FO-assignments, homomorphisms as well as regarding the determinant are shown.
Finally, we conclude in Section~\ref{sec:conclusion}.

\noindent Due to space limitations, all proof details can be found in the appendix.

\section{Preliminaries}\label{sec:prelims}

In this section, we describe the computational models and complexity classes that are relevant for parameterised complexity theory. 
We use standard notions and notations from parameterised complexity theory \cite{DBLP:series/mcs/DowneyF99,DBLP:series/txtcs/FlumG06}. 
Without loss of generality, we restrict the input alphabet to be $\{0,1\}$. 

\subparagraph*{Turing Machines (TMs) with Random Access to the Input}
We consider an intermediate model between TMs and Random Access Machines (RAMs) on words. 
Particularly, we make use of TMs that have random access to the input tape and can query relations in input structures in constant time.
This can be achieved with two additional tapes of logarithmic size (in the input length), called the \emph{random access tape} and the \emph{relation query tape}. 
On the former, the machine can write the index of an input position to get the value of the respective bit of the input. 
On the relation query tape, the machine can write a tuple $t$ of the input structure together with a relation identifier $R$ to get the bit stating whether $t$ is in the relation specified by $R$.
Note that our model achieves linear speed-up for accessing the input compared to the standard TM model. 
(This is further justified by Remark~\ref{remark:tail-ram}.) 
For convenience, in the following, whenever we speak about TMs we mean the TM model with random access to the input.
Denote by $\SPACETIME{s}{t}$ ($\NSPACETIME{s}{t}$) with $s,t\colon\N\to\N$ the class of languages that are accepted by (nondeterministic) TMs with space-bound $O(s(n))$ and time-bound $O(t(n))$. 
A $\calc$-machine for $\calc=\SPACETIME{s}{t}$ ($\calc=\NSPACETIME{s}{t}$) is a (nondeterministic) TM that is $O(s(n))$ space-bounded and $O(t(n))$ time-bounded.

NTMs are a generalisation of TMs where multiple transitions from a given configuration are allowed. 
This can be formalised by allowing the transition to be a relation rather than a function. 
Sometimes, it is helpful to view NTMs as deterministic TMs with an additional tape, called the (nondeterministic) choice tape which is read-only.
Let $M$ be a deterministic TM with a choice tape. 
A nondeterministic step in the computation of $M$ is a step where $M$ moves the head on the choice tape to a cell that was not visited before.
The \emph{language accepted by $M$}, $L(M)$ is defined as
\[
\{\,x \in \{0,1\}^*\mid\exists y\in \{0,1\}^*\text{ s.t.\ } M \text{ accepts } x\text{ when the choice tape is initialised with }y\,\}.
\]
Notice that in this framework the machine $M$ has two-way access to the choice tape.
Furthermore, resource bounds are with respect to the input only (the content of the choice tape is not part of the input) and the choice tape is not counted for space bounds.
In this paper, we regard nondeterministic TMs as deterministic ones with a choice tape.

Before we proceed to the definition of parameterised complexity classes, a clarification on the choice of the model is due. 
Note that RAMs and NRAMs are often appropriate in the parameterised setting as exhibited by several authors (see, e.g., the textbook of Flum and Grohe~\cite{DBLP:series/txtcs/FlumG06}). 
They allow to define bounded nondeterminism quite naturally. 
On the other hand, in the classical setting, branching programs (BPs) are one of the fundamental models that represent space bounded computation, in particular logarithmic space. 
Since BPs inherently use bit access, this relationship suggests the use of a bit access model.
Consequently, we consider a hybrid computational model: Turing machines with random access to the input.
While the computational power of this model is the same as that of Turing machines and RAMs, it seems to be a natural choice to guarantee a certain robustness, allowing for desirable characterisations of our classes.

\subparagraph*{Parameterised Complexity Classes} 
Let $\FPT$ denote the set of parameterised problems that can be decided by a deterministic TM running in time $f(k)\cdot p(|x|)$ for any input $(x,k)$ where $f$ is a computable function and $p$ is a polynomial.
Two central classes in parameterised complexity theory are $\W{1}$ and $\W\P$ which were originally defined via special types of circuit satisfiability~\cite{DBLP:series/txtcs/FlumG06}. 
Flum, Chen and Grohe~\cite{DBLP:conf/coco/ChenF03} obtained a characterisation of these two classes using the following notion of $k$-bounded NTMs. 
\begin{definition}[$k$-bounded TMs]
An NTM $M$, working on inputs of the form $(x,k)$ with $x \in \{0,1\}^*, k \in \mathbb{N}$, is said to be \emph{$k$-bounded} if for all inputs $(x,k)$ it reads at most $f(k)\cdot \log |x|$ bits from the choice tape on input $(x,k)$, where $f$ is a computable function.
\end{definition} 

Here, we will work with the following characterisation of $\W\P$.
The characterisation for $\W1$ needs another concept that will be defined on the next page.
\begin{proposition}[\cite{DBLP:conf/coco/ChenF03,DBLP:series/txtcs/FlumG06}]
$\W\P$ is the set of all parameterised problems that can be accepted by $k$-bounded $\FPT$-machines with a choice tape.
\end{proposition}

Now, we recall three complexity theoretic operators that define parameterised complexity classes from an arbitrary classical complexity class, namely $\para, \paraw$ and $\parab$, following the notation of Stockhusen~\cite{DBLP:phd/dnb/Stockhusen17}.

\begin{definition}[\cite{DBLP:journals/iandc/FlumG03}]
 Let $\calc$ be any complexity class. 
 Then $\para\calc$ is the class of all parameterised problems $P \subseteq \{0,1\}^* \times \mathbb{N}$ for which  there is a computable function $π \colon \mathbb{N} \to \{0,1\}^*$ and a language $L \in \calc$ with $L \subseteq \{0,1\}^* \times \{0,1\}^*$ \stfa $x \in\{0,1\}^*$, $k \in \mathbb{N}$:
 $(x,k) \in P \Leftrightarrow (x, π(k)) \in L.$
\end{definition}
Notice that $\para\P = \FPT$ is the standard precomputation characterisation of $\FPT$~\cite{DBLP:journals/iandc/FlumG03}. A $\para\calc$-machine for $\calc=\SPACETIME{s}{t}$ ($\calc=\NSPACETIME{s}{t}$) is a (nondeterministic) TM, working on inputs of the form $(x,k)$, that is $O(s(|x|+f(k)))$ space-bounded and $O(t(|x|+f(k)))$ time-bounded where $f$ is a computable function.

The class $\XP$ (problems accepted in time $|x|^{f(k)}$ for a computable function $f$) and the $\W{}$-hierarchy~\cite{DBLP:series/txtcs/FlumG06} capture intractability of parameterised problems.
Though the $\W{}$-hierarchy was defined using the weighted satisfiability of formulas with bounded weft, which is the number of alternations between gates of high fan-in, Flum and Grohe~\cite{DBLP:journals/iandc/FlumG03} characterised central classes in this context using bounded nondeterminism. 
Stockhusen and Tantau~\cite{DBLP:conf/iwpec/StockhusenT13,DBLP:phd/dnb/Stockhusen17} considered space-bounded and circuit-based parallel complexity classes with bounded nondeterminism. 

The following definition is a more formal version of the one given by Stockhusen and Tantau~\cite[Def.~2.1]{DBLP:conf/iwpec/StockhusenT13}.
They use $\text{para}\exists^{\leftrightarrow}_{f\log}\calc$ instead of $\paraw\calc$ for a complexity class $\calc$.
\begin{definition}
Let $\calc=\SPACETIME{s}{t}$ for some $s,t\colon\mathbb N\to\mathbb N$. 
Then, $\paraw\calc$ is the class of all parameterised problems $Q$ that can be accepted by a $k$-bounded $\para\calc$-machine with a choice tape.
\end{definition}

For example, $\paraw\L$ denotes the parameterised version of $\NL$ with $k$-bounded nondeterminism. 
One can also restrict this model by only giving one-way access to the nondeterministic tape. 
The following definition is a more formal version of the one of Stockhusen and Tantau~\cite[Def.~2.1]{DBLP:conf/iwpec/StockhusenT13} who use the symbol $\text{para}\exists^{\to}_{f\log}$ instead.
\begin{definition}
Let $\calc=\SPACETIME{s}{t}$ for some $s,t\colon\mathbb N\to\mathbb N$. 
Then $\parab\calc$ denotes the class of all parameterised problems $Q$ that can be accepted by a $k$-bounded $\para\calc$-machine with a choice tape with one-way read access to the choice tape.
\end{definition}
As there is only read-once access to the nondeterministic bits, $\parab\calc$ can be equivalently defined via nondeterministic transitions and without using a choice tape.

Another notion studied in parameterised complexity is tail-nondeterminism.
A $k$-bounded machine $M$ is \textit{tail-nondeterministic} if there exists a computable function $g$ such that on all inputs $(x,k)$, $M$ makes at most $g(k)\cdot \log n$ further steps in the computation, after its first nondeterministic step. 
The value of this concept is evidenced by the machine characterisation of $\W{1}$~(Chen~et~al.~\cite{DBLP:conf/coco/ChenF03}). 
We hope to get new insights by transferring this concept to space-bounded computation.
In consequence, we introduce the tail-nondeterministic versions of $\paraw\calc$ and $\parab\calc$ which are denoted by $\parawt\calc$ and $\parabt\calc$. 
 
\begin{remark}\label{remark:tail-ram}
Note that it is important to have random access to the input tape in the case of tail-nondeterminism.
Without random access to input bits and input relations, a TM cannot even make reasonable queries to the input in time $g(k)\cdot\log(n)$.
\end{remark}

\subparagraph*{Logic}
We assume basic familiarity with first-order logic (FO).
A \emph{vocabulary} is a finite ordered set of relation symbols and constants. 
Each relation symbol $R$ has an associated \emph{arity} $\arity R\in\mathbb N$.
Let $\tau$ be a vocabulary. A \emph{$\tau$-structure} $\struc A$ consists of a nonempty finite set $\dom(\struc A)$ (its \emph{universe}), and an \emph{interpretation} $R^{\struc A}\subseteq \dom(\struc A)^{\arity R}$ for every relation symbol $R\in\tau$.
Syntax and semantics are defined as usual (see, e.g., the textbook of Ebbinghaus~et~al.~\cite{DBLP:books/daglib/0080659}).
Let $\struc A$ be a structure with universe $A$.
We denote by $|\struc A|$ the \emph{size of a binary encoding of $\struc A$}, \ie, the number of bits required to represent the universe and  relations as lists of tuples. 
For example, if $R$ is a relation of arity $3$, then  $R^\struc{A}$ is represented as a subset of $A^3$, \ie, a set of triples over $A$. 
This requires $O(|R^{\struc A}|\cdot\arity{R})\cdot\log|A|)$ bits to represent the relation $R^{\struc A}$, assuming $\log |A|$ bits to represent an element in $A$.    
As analysed by Flum~et~al.~\cite[Sect.~2.3]{DBLP:journals/jacm/FlumFG02}, this means that $|\struc A|\in\Theta((|A|+|\tau|+\sum_{R\in\tau} |R^{\struc A}|\cdot\arity{R})\cdot\log|A|)$. 
Also recall that the fragment $\Sigma_i$ (for $i \in \mathbb{N}$) refers to the class of \FO-formulas with $i$ quantifier blocks alternating between existential and universal quantifiers and the outermost quantifier being existential.

\section{Parameterised Counting in Logarithmic Space}
\label{sec:class}
Now, we define the counting counterparts based on the parameterised classes defined using bounded nondeterminism.
The definitions of the decision classes based on tail-nondeterminism can be found in the appendix.
A \emph{parameterised function} is a function $F\colon\{0,1\}^* \times \mathbb{N} \to \mathbb{N}$. For an input $(x,k)$ of $F$ with $x \in \{0,1\}^*$, $k \in \mathbb{N}$, we call $k$ the \emph{parameter} of that input. 
If $\calc$ is a complexity class and a parameterised function $F$ belongs to $\calc$, we say that $F$ is $\calc$-computable.
Let $M$ be a TM.
We denote by $\acc_M(x)$ the number of accepting paths of $M$ on input $x$, and similarly, $\acc_M(x,k)$, for parameterised inputs of the form $(x,k)$.
\begin{definition}
Let $\calc=\SPACETIME{s}{t}$ for some $s,t\colon\mathbb N\to\mathbb N$. 
Then a parameterised function $F$ is in $\sh\paraw\calc$ if there is a $k$-bounded nondeterministic $\para\calc$-machine $M$ such that for all inputs $(x,k)$, we have that $\acc_M(x,k)=F(x,k)$. Furthermore, $F$ is in
\begin{itemize}
	\item $\sh\parab\calc$ if there is such an $M$ with read-once access to its nondeterministic bits,
	\item $\sh\parawt\calc$ if there is such an $M$ that is tail-nondeterministic, and
	\item $\sh\parabt\calc$ if there is such an $M$ with read-once access to its nondeterministic bits that is tail-nondeterministic.
\end{itemize}
\end{definition} 
By definition, we get $\sh\parabt\L\subseteq\calc\subseteq\sh\paraw\L$ for $\calc\in\{\sh\parab\L,\sh\parawt\L\}$. 
Note that the restriction of the above classes to $k$-boundedness is crucial.
If we drop this restriction, the machines are able to access $2^{f(k)+log|x|}$, so fpt-many, nondeterministic bits. 
Regarding multiple-read access, this allows for solving SAT (with constant parameterisation). 
So this class then would contain a $\para\NP$-complete problem.
For read-once access, we expect a similar result for $\para\NL$.
When adding tail-nondeterminism, we implicitly get $k$-boundedness again, so this does not lead to new classes. 

The following lemma shows that $\para\L$-machines can be normalised in a specific way.
This normalisation will play a major role in Section~\ref{sec:complete}.

\begin{lemma}\label{normalise-lemma}
	For any $k$-bounded nondeterministic $\para\L$-machine $M$ there exists a $k$-bounded nondeterministic $\para\L$-machine $M'$ with $\sh\acc_M(x,k) = \sh\acc_{M'}(x,k)$ for all inputs $(x,k)$ such that $M'$  has the following properties:
	\begin{enumerate}[(1)]
        \item $M'$ has a unique accepting configuration,
        \item  on any input $(x,k)$, every computation path of $M'$ accesses exactly $g(k)\cdot \log |x|$ nondeterministic bits (for some computable function $g$), and $M'$ counts on an extra tape (\tapeS) the number of nondeterministic steps, and
        \item $M'$ has an extra tape (\tapeC) on which it remembers previous nondeterministic bits, resetting the tape after every $\log|x|$-many nondeterministic steps.
	\end{enumerate}
	Additionally, if $M$ has read-once access to its nondeterministic bits, or is tail-nondeterministic, or both, then $M'$ also has these properties.
\end{lemma}

The following result follows from a simple simulation of nondeterministic machines by deterministic ones.
Let $\F\FPT$ be the class of functions computable by $\FPT$-machines with output.
\begin{theorem}
\label{thm:fpt-ub}
	$\sh\parab\L\subseteq\F\FPT$.
\end{theorem}

Using the notion of oracle machines (see, e.g.,~\cite{DBLP:books/daglib/0086373}), we define Turing, metric, and parsimonious reductions computable in $\para\L$.
For our purposes, the oracle tape is always exempted from space restrictions which is often the case in the context of logspace Turing reductions~\cite{DBLP:journals/jcss/Buss88}. A study on the effect of changing this assumption might be interesting.
\begin{definition}[Reducibilities]\label{def:reducibilities}
Let $F, F' \colon \{0,1\}^* \times \mathbb{N} \to \mathbb{N}$ be two functions. 
Then, $F$ is \emph{para-logspace Turing reducible} to $F'$, $F \leqparaturing F'$, if there is a $\para\L$ oracle TM $M$ that computes $F$ with oracle $F'$ and the parameter of any oracle query of $M$ is bounded by a computable function in the parameter. 
If there is such an $M$ that uses only one oracle query, then $F$ is \emph{para-logspace metrically reducible} to $F'$, $F \leqparalogmet F'$. 
If there is such an $M$ that returns the answer of the first oracle query, then $F$ is \emph{para-logspace parsimoniously reducible} to $F'$, $F \leqparalogpars F'$.
\end{definition}
Note that the definition of parsimonious reductions ensures that the size of the witness set is preserved by the fact that $M$ immediately returns the answer of its only oracle query (without further computations).
For any reducibility relation $\preccurlyeq$ and any complexity class $\calc$, $[\calc]^\preccurlyeq \dfn\{\,A\mid\exists C\in\calc\text{ \ST\ }A\preccurlyeq C \,\}$ is the \emph{$\preccurlyeq$-closure of $\calc$}.

Next, we show that both new classes without tail-nondeterminism are closed under $\leqparalogpars$.
\begin{lemma}\label{lem:closure-legparalogpars-reductions}
	The classes $\sh\paraw\L$ and $\sh\parab\L$ are closed under $\leqparalogpars$.
\end{lemma}
For the tail-classes, such a closure property is not obvious. 
Corollary~\ref{thm:reach-parabetataill-complete} will show that closing the class with read-once access and tail-nondeterminism under these reductions gives the full power of the class without tail-nondeterminism. 
Open Problem~\ref{op:MC} on page~\pageref{op:MC} asks what class is obtained when closing the class without read-once access and with tail-nondeterminism.

Another important question is whether classes are closed under certain arithmetic operations. 
We show that all newly introduced classes are closed under addition and multiplication. 
\begin{theorem}
\label{prop:closure-log}
For any $o\in\{\W{},\W1,\beta,\beta\mbox{-}\complClFont{tail}\}$, the class $\sh\complClFont{para}_o\mbox-\L$ is closed under addition and multiplication.
\end{theorem}

\begin{openproblem}\label{op:monus}
	Which of the classes are closed under \emph{monus}, that is, $\max\{F-G,0\}$?
\end{openproblem}

\section{Complete Problems}
\label{sec:complete}
This section studies complete problems for the previously defined classes:  
counting problems in the context of walks in directed graphs, model-checking problems for $\FO$ formulas, and homomorphisms between \FO-structures as well as a parameterised version of the determinant.

\subsection{Counting Walks}\label{ssec:countingpaths}
We start with parameterised variants of counting walks in directed graphs which will be shown to be complete for the introduced classes. 

\paracountingproblemdef{$\p\sh\LOG\reach\bdegree$}{directed graph $\G=(V,E)$ with out-degree $b$, $s, t \in V$ and $a,k\in\N$}{$k$}{number of $s$-$t$-walks of length $a$ if $a \leq k\cdot\log |V|$, 0 otherwise}

\begin{theorem}\label{thm:logreach-parabetaL-paralog-complete}
	For every $b\ge2$, $\p\sh\LOG\reach\bdegree$ is $\sh\paraBetaL$-complete \wrt $\leqparalogpars$-reductions.
\end{theorem}
\begin{proof}[Proof Idea.]
	For the upper bound, guess a path of length exactly $a$.
	The number of nondeterministic bits is bounded by $O(k\cdot\log|V|)$ since successors can be referenced by a number in $\{0,\dots,b-1\}$.

	For the lower bound, using Lemma~\ref{normalise-lemma}, construct the configuration graph $\G$ restricted to nondeterministic configurations and the unique accepting configuration $C_\acc$, where the edge relation expresses whether a configuration is reachable with exactly one nondeterministic, but an arbitrary number of deterministic steps.
	Accepting computations of the machine correspond to paths from the first nondeterministic configuration to $C_\acc$ of length $f(k)\cdot\log|\G|$ in $\G$.
\end{proof}

Now consider the problem $\p\sh\reach$, defined as follows. 
\paracountingproblemdef{$\p\sh\reach$}{directed graph $\G=(V,E)$, $s, t \in V$, $k\in\N$}{$k$}{number of $s$-$t$-walks of length exactly $k$}
The difference to the previous problem is the unbounded out-degree of nodes and the length of counted walks being $k$ instead of $a \leq k \cdot \log|x|$.
Note that the analogue problem for counting paths is $\#\W1$-complete \cite{DBLP:journals/siamcomp/FlumG04}.
However, we will see now that the problem for walks is $\sh\paraBetaL$-complete.

\begin{theorem}\label{thm:reach-parabetal-complete}
	$\p\sh\reach$ is $\sh\paraBetaL$-complete \wrt $\leqparalogpars$.
\end{theorem}

As the length of paths that are counted in $\p\sh\reach$ is $k$, the runtime of the whole algorithm used to prove membership in the previous theorem is actually bounded by $k \cdot \log |x|$ on input $(x,k)$.
This means that the computation is tail-nondeterministic.

\begin{theorem}\label{thm:reach-parabetataill-complete}
  $\p\sh\reach$ is $\sh\parabt\L$-complete \wrt $\leqparalogpars$.
\end{theorem}

The previous results together with the fact that $\sh\paraBetaL$ is closed under $\leqparalogpars$ yield the following surprising collapse (a similar behaviour was observed by Bottesch~\cite{DBLP:conf/iwpec/Bottesch17,DBLP:conf/mfcs/Bottesch18}).

\begin{corollary}\label{cor:parabtaillog-clos=parabl}
  $\clospara{\sh\parabt\L} = \sh\paraBetaL$.
\end{corollary}

We continue with another variant of $\p\sh\LOG\reach\bdegree$, namely $\p\sh\LOG\WALK\bdegree$. Here, all walks of length $a$ are counted, if $a \leq k \cdot \log |x|$ (and $s,t$ are not part of the input).

\begin{theorem}\label{thm:pathb-parabetal-complete-turingplog}
	$\p\sh\LOG\WALK\bdegree$ is $\sh\paraBetaL$-complete \wrt $\leqparaturing$.
\end{theorem}

Now, consider a problem that combines a reachability problem with model-checking for propositional logic, that is, it only counts walks that are models of a propositional formula (see Haak~et~al.~\cite{DBLP:conf/mfcs/HaakKMVY19}).
Let $\G=(V,E)$ be a DAG, $(e_1,\dots,e_n)$ be a walk in $\G$ with $e_i\in E$ for $1\leq i\leq n$, and $P=\{e_1,\dots,e_n\}$.
Define the function $c_P \colon E \to \{0,1\}$ to be the characteristic function of $P$ \wrt $E$: $c_P(e)=1$ iff $e \in P$.
\paracountingproblemdef{$\p\sh\LOG\reach_2\CNF$}{directed graph $\G=(V,E)$ of out-degree 2, $s,t \in V$, CNF formula $\varphi$ with $\Vars\varphi\subseteq E$, $a,k\in\N$}{$k$}{Number of $s$-$t$-walks $(s=e_1, \dots, e_{a}=t)$ \ST $c_P\models\varphi$, where $P = \{e_1, \dots, e_a\}$, if $a \leq k\cdot\log(|V|+|φ|)$, 0 otherwise}

\begin{theorem}\label{thm:reach2cnf-parawl-complete}
	$\p\sh\LOG\reach_2\CNF$ is $\sh\paraWL$-complete \wrt $\leqparalogpars$.
\end{theorem}
\begin{proof}[Proof Idea.]
	Regarding membership, we can first use the algorithm outlined in the proof idea of Theorem~\ref{thm:logreach-parabetaL-paralog-complete} to nondeterministically guess a path, and then use the standard logspace model-checking algorithm for propositional formulas.
	Since edges in the graph are associated with variables of the formula, whenever we need the value of an edge variable $e$, we run the original algorithm re-using nondeterministic bits to determine it.

	For the lower bound, we use the  same graph  as in   Theorem~\ref{thm:logreach-parabetaL-paralog-complete}, and the formula is used to express consistency of the re-used nondeterministic bits in the configuration graph.
\end{proof}


Similarly, define the problem $\p\sh\cyclecover_2\CNF$: Given a graph $\G = (V,E)$ of  bounded   out-degree 2, a CNF-formula $\varphi$ with $\Vars{\varphi} \subseteq E$ and $a, k \in \mathbb{N}$, with $k$ as the parameter and $a\leq \log(|G| + |\varphi|)$, output the number of cycle covers $D \subseteq E$ in which the number of non-selfloop-cycles is $\leq k$, exactly $k \cdot a$ vertices are covered non-trivially and $\Vars{D} \models \varphi$. 

\begin{theorem}\label{thm:cyclecover-parawl}
	$\p\sh\cyclecover_2\CNF$ is $\sh\paraWL$-complete \wrt $\leqparalogpars$.
\end{theorem}

\subsection{Counting FO-Assignments}
Let $\class{F}$ be a class of well-formed formulas.  
The problem of counting satisfying assignments to free \FO-variables in $\class{F}$-formulas, $p\text-\sh\MC{\class{F}}$, is defined as follows.
\paracountingproblemdef{$p\text-\sh\MC{\class{F}}$}{formula $\varphi\in \class{F}$, structure $\struc A$, $k\in\N$}{$|\varphi|$}{$|\varphi(\struc A)|$ if $k=|\varphi|$, 0 otherwise}
Here, $φ(\struc A)$ is the set of satisfying assignments of φ in $\struc A$: $\varphi(\struc A)  = \{\,(a_1,\ldots, a_n)\mid (a_1,\ldots, a_n) \in \dom(\struc A)^n, {\struc A}\models \varphi(a_1,\ldots, a_n)\,\}$. 
%
Denote by $p\text-\sh\MC{\class{F}}_a$ the variant where for all relations the arity is at most $a\in\mathbb N$.
We investigate parameterisations that yield complete problems for some of the new classes in this setting.

In particular, we consider a fragment of $\FO$ obtained by restricting the occurrence of variables in the syntactic tree of a formula in a purely syntactic manner.
Formally, the \emph{syntax tree} of a quantifier-free $\FO$-formula $\varphi$ is a tree with edge-ordering whose leaves are labelled by atoms of $\varphi$ and whose inner vertices are labelled by Boolean connectives.
\begin{definition}
Let $r \in \mathbb{N}$ and φ be a quantifier-free $\FO$-formula. Let $θ_1, \dots, θ_m$ be the atoms of φ in the order of their occurrence in the order-respecting depth-first run through the syntax tree of φ. We say that φ is \emph{$r$-local} if for any $θ_i$, $θ_j$ that involve the same variable, we have $|i-j| \leq r$.
We define 
	$\Sigmarlocal\dfn\{\,\varphi\in\Sigma_0\mid \varphi\text{ is $r$-local}\,\}.$
\end{definition}

Using this syntactic notion, we obtain a complete problem for our classes with read-once access to nondeterministic bits in the setting of first-order model-checking.
\begin{theorem}\label{thm:mc-local-bounded}
	For $a\ge2,r\ge1$, $\p\sh\MC{\Sigmarlocal}_a$ is $\sh\paraBetaL$-complete and $\sh\parabt\L$-complete \wrt $\leqparalogpars$.
\end{theorem}
\begin{proof}[Proof Idea.]
	Regarding membership, we evaluate the given $\varphi$ in $\struc A$ top to bottom using the locality of $\varphi$ by storing assignments to variables until we encountered $r$ more atoms.
	As a result, at most $a\cdot r$ assignments to variables are simultaneously stored and each one needs $\log|\struc A|$ space.
	Moreover, the runtime of the whole procedure is bounded by $f(|\varphi|)\cdot\log|\struc A|$ for some computable function $f$ and thereby the procedure is tail-nondeterministic.

	Regarding the lower bound, we reduce from $\p\sh\reach$ and use the formula
	\[φ_k(x_1, \dots, x_k) \dfn (x_1 = s) ∧ E(x_1, x_2) ∧ E(x_2, x_3) ∧ \dots ∧ E(x_{k-1}, x_k) ∧ x_k = t\]
	expressing that a tuple of vertices $(v_1, \dots, v_k)$ is an $s^{\struc{A}}$-$t^{\struc{A}}$-walk in an $(E,s,t)$-structure~$\struc{A}$.
\end{proof}

Note that the decision version of $p\text-\sh\MC{\Sigma_0}$ is equivalent to parameterised model-checking for $\Sigma_1$-sentences, as we count assignments to free variables.
This problem characterises tail-nondeterministic para-logspace with read-once access to nondeterministic bits.
%
%
%
%

\begin{theorem}
\label{thm:mc-for-w1}
	$p\text-\sh\MC{\Sigma_0}$ is $\sh\parawt\L$-complete \wrt $\leqparalogpars$.
\end{theorem}

The complexity status of counting assignments to free first-order variables in a $\Sigma_0$ formula with unbounded arity or without the local restrictions is not known. 
In particular, it is not clear if the restriction on the arity or the locality property of the formula can be removed while preserving completeness. 
Finally, we close this section with three open questions. 

\begin{openproblem}\label{op:MC}
	What is the complexity of $\clospara{\p\sh\MC{\Sigma_0}_a}$ for fixed $a\in\N$?
	What is the complexity of $\clospara{\p\sh\MC{\Sigmarlocal}}$ for fixed $r\in\N$?
\end{openproblem}

\begin{openproblem}\label{op:paraw1l-paralogpars}
	Is the class $\clospara{\sh\parawt\L}$ equivalent to some known class?
\end{openproblem}


\subsection{Counting Homomorphisms}
This subsection is devoted to the study of the problem of counting homomorphisms between two structures in the parameterised setting.
Typically, the size of the universe of the first structure is considered as the parameter.
The complexity of counting homomorphisms has been intensively investigated for almost two decades \cite{DBLP:journals/rsa/DyerG00,DBLP:journals/jacm/Grohe07,DBLP:journals/tcs/DalmauJ04,DBLP:journals/toct/ChenM15}.

\begin{definition}[Homomorphism]
	Let $\struc A$ and $\struc B$ be structures over some vocabulary $\tau$ with universes $A$ and $B$, respectively.
	A function $h\colon A\to B$ is a \emph{homomorphism from $\struc A$ to $\struc B$} if for all $R\in\tau$ and for all tuples $(a_1,\dots, a_{\arity{R}})\in R^{\struc A}$, we have $(h(a_1),\dots,h(a_{\arity{R}})) \in R^{\struc B}$.
\end{definition}

A bijective homomorphism $h$ between two structures $\struc A,\struc B$ such that the inverse of $h$ is also a homomorphism is called an \emph{isomorphism}.
If there is an isomorphism between $\struc A$ and $\struc B$, then $\struc A$ is said to be \emph{isomorphic} to~$\struc B$.

\begin{definition}
	Let $\struc A$ be a structure with universe $A$. 
	Then denote by $\struc A\!^*$ the extension of $\struc A$ by a fresh unary relation symbol $C_a$ interpreted as $C_a^{\struc A}=\{a\}$ for each $a\in \dom(\struc A)$.
	Analogously, denote by $\class{A}^*$ for a class of structures $\class{A}$ the class $\{\,\struc A\!^*\mid\struc A\in\class{A}\,\}$.
\end{definition}

Define $\pcountHom{\class{A}}$ as the following problem. Given a pair of structures $(\struc A,\struc B)$ where $\struc A\in\class{A}$, and parameter $k$, output the number of homomorphisms from $\struc A$ to $\struc B$, if $|\dom(\struc A)| \leq k$, and 0 otherwise.
\paracountingproblemdef{$\pcountHom{\class{A}}$}{A pair of structures $(\struc A,\struc B)$ where $\struc A\in\class{A}$}{$|\struc A|$, $k\in\N$}{the number of homomorphisms from $\struc A$ to $\struc B$ if $|\dom(\struc A)|\leq k$, 0 otherwise}
Notice that $\struc B$ can be any structure.
For $n \geq 2$, let $P_n$ be the canonical undirected path of length $n$, that is, the $(E)$-structure with universe $\{1, \dots, n\}$ and $E^{P_n} = \{\,(i,i+1),(i+1,i)\mid 1 \leq i < n\,\}$.
Let $\class{P}$ be the class of structures isomorphic to some $P_n$.
For the next theorem, reduce to $\p\sh\reach$ for membership, and from a normalised, coloured variant of $\p\sh\reach$ for hardness.

\begin{theorem}\label{thm:phom(p^*)-pbetaL-complete}
	$\pcountHom{\class{P}^*}$ is $\sh\paraBetaL$-complete \wrt $\leqparalogpars$.
\end{theorem}

\begin{openproblem}\label{op:phom-natural-w1L-complete}
	Is there a natural class of structures $\class A$ \ST $\pcountHom{\class{A}}$ is $\sh\parawt\L$-complete \wrt $\leqparalogpars$?
\end{openproblem}

\subsection{The Parameterised Complexity of the Determinant}
\label{subsec:det}
In this section, we consider a parameterised variant of the determinant function introduced by Chauhan and Rao~\cite{DBLP:conf/caldam/ChauhanR15}.
For $n>0$ let $\class{S}_n$ denote the set of all permutations of $\{1,\ldots, n\}$. 
For $k\le n$, let ${ \class{S}}_{n,k}$ denote the following subset of $\class{S}_n$: 
${\class {S}}_{n,k} = \{\,\pi\,|\, \pi\in S_n, |\{\,i\,:\,\pi(i) \neq i\,\}| = k\,\}. $

We define the parameterised determinant function of an $n\times n $ square matrix $A = (a_{i,j})_{1 \le i,j \le n}$ as $\pdet(A,k) = \sum_{\pi \in S_{n,k}} \sign{\pi} \prod_{i: \pi(i) \neq i} a_{i,\pi(i)}$.

Using an interpolation argument, it can be shown that $\pdet$ is in $\F\P$ when $k$ is part of the input and thereby in $\F\FPT$~\cite{DBLP:conf/caldam/ChauhanR15}, the functional counterpart of $\FPT$.
In fact, the same interpolation argument can be used to show that $\pdet$ is in ${\gap\L}$ (the class of functions $f(x)$ \ST for some $\NL$-machine, $f(x)$ is the number of accepting minus the number of rejecting paths). 
However, this does not give a space efficient algorithm for $\pdet$ in the sense of parameterised classes. 
The $\gap\L$ algorithm may require a large number of nondeterministic steps and accordingly is not $k$-bounded.  
We show that the space efficient algorithm for the determinant given by Mahajan and Vinay~\cite{DBLP:journals/cjtcs/MahajanV97} can be adapted to the parameterised setting, proving that $\pdet$ can be written as a difference of two $\sh\parab\L$ functions. 
Recall the notion of a \emph{clow sequence} introduced by Mahajan and Vinay~\cite{DBLP:journals/cjtcs/MahajanV97}. 
\begin{definition}[Clow]
Let $\G=(V,E)$ be a directed graph with $V = \{1,\ldots, n\}$ for some $n \in \mathbb{N}$. 
A {\em clow} in $\G$ is a walk $C= (w_1,\ldots, w_{r-1}, w_r = w_1)$ 
where $w_1$ is the minimal vertex among $w_1, \dots, w_{r-1}$ \wrt the natural ordering of $V$ and $w_1 \neq w_j$ for all $1 < j < r$. Node $w_1$ is called the \emph{head} of $C$, denoted by $\head(C)$.
\end{definition}
\begin{definition}[Clow sequence]
	A {\em clow sequence} of a graph $\G=(\{1,\dots,n\},E)$ is a sequence ${W} = (C_1,\ldots, C_k)$ such that $C_i$ is a clow of $\G$ for $1 \leq i \leq k$ and
\begin{itemize}
  \item the heads of the sequence are in ascending order $\head(C_1) <\cdots < \head(C_k)$, and
  \item the total number of edges that appear in some $C_i$ (including multiplicities) is exactly $n$.
\end{itemize}  
\end{definition}

For a clow sequence $W$ of some graph $\G=(\{1,\dots,n\}, E)$ with $r$ clows the \emph{sign of $W$}, $\sign{W}$, is defined as $(-1)^{n+r}$. 
Note that, if the clow sequence is a cycle cover $\sigma$, then $(-1)^{n+r}$ is equal to the sign of the permutation represented by $\sigma$ (that is, $(-1)^{\#\text{ inversions in }\sigma}$).
Mahajan and Vinay came up with this sign-function to derive their formula for the determinant. 

For an $(n\times n)$-matrix $A$, $\G_A$ is the weighted directed graph with vertex set $\{1, \dots, n\}$ and weighted adjacency matrix $A$. 
For a clow (sequence) $W$, $\wt{W}$ is the product of weights of the edges (clows) in $w$. 
For any $\G$ as above, $\class{W}_{\G}$ is the set of all clow sequences of $\G$. 
Mahajan and Vinay~\cite{DBLP:journals/cjtcs/MahajanV97} proved that 
$ \det(A) = \sum_{W\in {\class{W}_{\G_A}}} \sign{W}\cdot\wt{W}.$

We adapt these notions to the parameterised setting. First observe that for a permutation $\sigma \in S_{n,k}$, we have that $\sign{\sigma} = (-1)^{n+r}$, where $r$ is the number of cycles in the permutation. 
However, the number of cycles in $\sigma$ is $n-k+r'$, where $r'$ is the number of cycles of length at least two in $\sigma$. 
Accordingly, we have $\sign{\sigma} = (-1)^{2n-k+r'}$.  
Adapting the definition of a clow sequence, for $k \ge 0$, define a {\em $k$-clow sequence} to be a clow sequence where the total number of edges (including multiplicity) in the sequence is exactly $k$, every clow has at least two edges, and no self loop edge of the form $(i,i)$ occurs in any of the clows. 
For any graph $\G$ with vertex set $\{1,\dots,n\}$ for  $n \in \mathbb{N}$, ${\class{W}}_{\G,k}$ is the set of all $k$-clow sequences of $\G$. 
For a $k$-clow sequence $W \in {\class{W}}_{\G,k}$, $\sign{W}$ is $(-1)^{2n-k+r'}$, where $r'$ is the number of clows in $W$.  
Mahajan and Vinay~\cite[Theorem~1]{DBLP:journals/cjtcs/MahajanV97} showed that the signed sum of the weights of all clow sequences is equal to the determinant. 
At the outset, this is a bit surprising, since the determinant is equal to the signed sum of weights of cycle covers, whereas there are clow sequences that are not cycle covers. 
Mahajan and Vinay~\cite{DBLP:journals/cjtcs/MahajanV97} observed that every clow sequence that is not a cycle cover can be associated with a unique clow sequence of opposite sign, and thereby all clow sequences cancel out. 
We observe a parameterised version of the above result~\cite[Theorem~1]{DBLP:journals/cjtcs/MahajanV97}. 
\begin{lemma}\label{lem:clow}
$\pdet(A,k) = \sum_{W \in {\cal W}_{\G_A,k}}\sign{W}\cdot\wt{W}$, for $\{0,1\}$-matrix $A$, $k \in \mathbb{N}$.
\end{lemma}

Using this characterisation, the upper bound in the following theorem can be obtained. 
For hardness a reduction from $\p\sh\reach$ suffices.
\begin{theorem}
\label{thm:pdet-ub}
The problem $\pdet$ for $(0,1)$-matrices can be written as a difference of two functions in $\sh\parabt\L$, and is $\sh\parabt\L$-hard \wrt $\leqparalogmet$.
\end{theorem}

%

\section{Conclusions and Outlook}\label{sec:conclusion}
We developed foundations for the study of parameterised space complexity of counting problems. 
Our results show interesting characterisations for classes defined in terms of $k$-bounded para-logspace NTMs. 
We believe that our results will lead to further research of parameterised logspace counting complexity.
Notice, that the studied walk problems in Section~\ref{ssec:countingpaths} can be considered restricted to DAGs yielding the same completeness results.

Branching programs are immanent for the study of space-bounded and parallel complexity classes. 
Languages accepted by polynomial-size logspace-uniform branching programs characterise $\NL$. 
In fact, this result carries forward to the counting versions. 
Motivated by this, one can consider parameterised counting classes based on deterministic branching programs (DBPs) and nondeterministic branching programs (BPs). 
It can be shown that for any $o\in\{\W{},\W1,\beta,\beta\mbox{-}\complClFont{tail}\}$,  $\sh\complClFont{para}_o\mbox-\L$ and $\sh\complClFont{para}_o\mbox-\NL$, can be characterised in terms of an adequate parameterised counting version of DBPs and BPs, respectively (see Theorems~\ref{lem:equivalence-BP-SC} and~\ref{thm:bp-nl-equality} in the appendix). 

Comparing our newly introduced classes with the $\W{}$-hierarchy (which is defined in terms of weighted satisfiability problems for circuits of a so-called bounded weft), one might ponder whether there is an alternative definition of our classes with such circuit problems.
Though in this article we did not explore the weighted satisfiability, the closely related problem $p\text-\MC{\Sigma_0}$ sheds some light on this. 
Theorem~\ref{thm:mc-for-w1} shows that $p\text-\MC{\Sigma_0}$ is complete for $\parawt\L$ (in fact, we show this for their counting versions) under $\leq^{\mathrm{plog}}_{m}$-reductions. 
However, if we take $\FPT$-reductions, $p\text-\MC{\Sigma_0}$ is complete for $\W1$.
Though we could not prove it so far, we believe this is a general phenomenon: Any $\W1$-complete problem is complete for $\parawt\L$ under $\leq^{\mathrm{plog}}_{m}$-reductions. 
More generally, there is a possibility that the FPT-closure of $\paraWL$-classes is equal to the corresponding class in the $\W{}$-hierarchy. 

One might also ask the question if $\paraWL$ is contained in $\F\FPT$.
This is unlikely based on the view expressed above. 
For example, $p\text-\MC{\Sigma_0}$ is complete for both $\parawt\L$ and $\W1$ but under two different reductions. 
As a result, $\paraWL\subseteq\F\FPT$ would imply that $p\text-\MC{\Sigma_0}\in\FPT$ and, accordingly, $\FPT = \W1$ as $\FPT$ is closed under $\FPT$-reductions.

\noindent We close with interesting tasks for further research:
\begin{itemize}
	\item Study further closure properties of the new classes (e.g., Open Problem~\ref{op:monus}). 
	\item Improve the understanding of the influence of syntactic locality, resp., bounded arity in the setting of $\p\sh\MC{\Sigma_0}$ (Open Problem~\ref{op:MC}).
	\item Find a characterisation of the $\leqparalogpars$-closure of $\sh\parawt\L$ (Open Problem~\ref{op:paraw1l-paralogpars}).
	\item Identify a natural class of structures for which the homomorphism problem is $\#\paraw\L$-complete (Open Problem~\ref{op:phom-natural-w1L-complete}).
	\item Establish a broader spectrum of complete problems for the classes $\paraBetaL$ and $\paraWL$, \eg, in the realm of satisfiability questions.
	\item Identify further characterisations of the introduced classes, \eg, in the vein of descriptive complexity, which could highlight their robustness.
	\item Study gap classes~\cite{DBLP:journals/jcss/FennerFK94} based on our classes. This might help improve Theorem~\ref{thm:pdet-ub}.
\end{itemize}

\bibliography{ref}
\appendix
\newpage
\section{Omitted Proof Details}
\begin{restatelemma}[normalise-lemma]
	\begin{lemma}
		For any $k$-bounded nondeterministic $\para\L$-machine $M$ there exists a $k$-bounded nondeterministic $\para\L$-machine $M'$ with $\sh\acc_M(x,k) = \sh\acc_{M'}(x,k)$ for all inputs $(x,k)$ such that $M'$  has the following properties:
		\begin{enumerate}[(1)]
			\item $M'$ has a unique accepting configuration,
			\item  on any input $(x,k)$, every computation path of $M'$    accesses exactly $g(k)\cdot \log |x|$ nondeterministic bits (for some computable function $g$), and $M'$ uses an extra tape (\tapeS) that counts the number of nondeterministic steps, and
			\item $M'$ has an extra tape (\tapeC) on which it remembers the last nondeterministic bits, resetting the tape after every $\log|x|$-many nondeterministic steps.
		\end{enumerate}
		Additionally, if $M$ has read-once access to its nondeterministic bits, or is tail-nondeterministic, or both, then $M'$ also has these properties.
	\end{lemma}
	\end{restatelemma}
	
	\begin{proof}
	  We construct the machine $M'$ with the three desired properties and the same number of accepting computation paths as $M$ step by step, ensuring that the properties from previous steps are preserved. 
		
		First note that \wLOG $M$ can be assumed to have a single accepting state.
		We can modify $M$ so that upon reaching an accepting state, it erases everything in the work tape and move the head positions of every tape to the left end marker. 
		This ensures that $M$ has a unique accepting configuration, which is property (1).
		This does not alter the number of accepting paths of $M$ on any input. 
	
		To ensure (2), we construct a new machine $N$ that behaves as $M$ but additionally maintains a counter on \tapeS for the access of nondeterministic bits: Every time a nondeterministic bit is accessed, the counter is incremented by $1$. 
		If $M$ halts with the counter being less than $g(k)\cdot\log|x|$, then the modified machine $N$ keeps making nondeterministic choices until the count is $g(k)\cdot\log|x|$. 
		We define $N$ to accept if and only if $M$ accepts and all of the additional nondeterministic bits guessed by $N$ have the value $0$. 
		Note that $N$ does not use any additional space except for maintaining and updating the counter.
	  For all $(x,k) \in \{0,1\}^* \times \mathbb{N}$, the machine $N$ accesses exactly $g(k)\cdot\log |x|$ nondeterministic bits on all computation paths and $\sh\acc_M(x,k) = \sh\acc_N(x,k)$. Also, $N$ still has property (1), since the tape $S$ has the same content and head position for all accepting configurations and the remaining part of accepting configurations doesn't change compared to $M$.
	 
		Finally, to ensure (3), we modify $N$ from the previous step to obtain $M'$ as follows.
		The new machine $M'$ has an additional tape (\tapeC) on which it initially marks exactly $\log|x|$ cells, placing the head on the left-most cell afterwards.
		Whenever $N$ reads a nondeterministic bit, $M'$ copies the nondeterministic bit to \tapeC (that is, $M'$ copies the bit to \tapeC and moves the head position to the right). 
		If the current head position in \tapeC is on the right-most marked cell, then $M'$ erases the content of \tapeC in the marked cells and copies the nondeterministic bit currently being read to the first marked cell on \tapeC.
		Finally, $M'$ accepts if and only if $N$ accepts. Since this modification does not introduce new nondeterministic steps, the number of accepting computation paths of $M'$ on any input is the same as that of $N$ and the modification preserves property (2). To re-establish property (1), we clear tape $C$ whenever we reach an accepting configuration.  
\end{proof}

\begin{restatetheorem}[thm:fpt-ub]
\begin{theorem}
	$\sh\parab\L\subseteq\F\FPT$.
\end{theorem}
\end{restatetheorem}
\begin{proof}
Let $F\in \sh\parab\L$ via the machine $M$ with space-bound $g(k) + c\log n$ for some constant $c>0$. 
For an input $(x,k)$, let $\G(x,k)=(V(x,k),E(x,k))$ be the configuration graph of $M$ on $(x,k)$. 
Since $|V(x,k)| = 2^{O(g(k) + c \log|x|)}$, $\G$ can be constructed in $\FPT$ time given the input $(x,k)$. 
Now, $\sh\acc_M(x,k)$ is equal to the number of paths from the starting configuration $s(x,k)$ of $M$ on input $(x,k)$ to the unique accepting configuration $C_\acc$ (by virtue of Lemma~\ref{normalise-lemma}).
This can be done via a brute-force algorithm computing the number of accepting paths in $\G(x,k)$, since the size of $\G(x,k)$ is fpt and the branching degree of the vertices is at most $2$.
That is, we systematically iterate through all paths starting from $s(x,k)$, storing the currently investigated path, and incrementing the count whenever we reach $C_\acc$.
\end{proof}

\begin{restatelemma}[lem:closure-legparalogpars-reductions]
	\begin{lemma}
	The classes $\sh\paraw\L$ and $\sh\parab\L$ are closed under $\leqparalogpars$.
	\end{lemma}
\end{restatelemma}
\begin{proof}
	We start with $\sh\paraw\L$.
	Let $F,{F'}\subseteq\{0,1\}^*\times\mathbb N$ and $F'\in\sh\paraw\L$ via the $k$-bounded, $s_{F'}$ space-bounded NTM $M_{F'}$.
	Furthermore, let $F\leqparalogpars {F'}$ via the $\para\L$ oracle TM $M$.
	Let $\sigma\colon\{0,1\}^*\times\N\to\{0,1\}^*$ and $h\colon\{0,1\}^*\times\N\to\N$ \ST the machine $M$ on input $(x,k)$ uses the only oracle query $(\sigma(x,k),h(x,k))$.
	
	Let $M_\sigma$, $M_h$ with space-bounds $s_\sigma$, $s_h$ be the machines computing $\sigma$, $h$. 
	To show that $F\in\sh\paraw\L$, we construct a new $\paraw\L$-machine $M_F$ as follows.
	
	On input $(x,k)$, the machine $M_F$ simulates $M_{F'}$ on $(\sigma(x,k),h(x,k))$ using $M_\sigma$ and $M_h$.
	Initially, $h(x,k)$ is computed using $M_h$ and the value is stored.
  Then, whenever $M_{F'}$ reads the $i$-th input symbol, $M_F$ runs $M_\sigma$ on $(x,k)$ until it outputs the $i$-th symbol and uses it as the next input symbol.
  For this, note that $|σ(x,k)|$ is fpt-bounded.
	Afterwards $M_F$ continues the simulation of $M_{F'}$.
	On $(x,k)$ the number of accepting paths of $M_F$ is exactly the number of accepting paths of $M_{F'}$ on $(\sigma(x,k),h(x,k))$ which is equal to ${F'}(\sigma(x,k),h(x,k))=F(x,k)$ as required.
	The space of $M_F$ is bounded by the space used by $M_{F'}$ on $(\sigma(x,k),h(x,k))$ and space required for running $M_\sigma$ and bookkeeping.
	Regarding bookkeeping, we need an index counter for $M_{F'}$'s input head position and the value $h(x,k)$.
	Formally, the space is bounded by
	\[
	s_h(x,k)+s_{F'}(|\sigma(x,k)|,h(x,k))+s_\sigma(|x|,k)+s_\text{bk}(|x|,k),
	\]
	where $s_\text{bk}(|x|,k)$ is the space required for bookkeeping.
	This sum is in $O(\log|x|+f(k))$.
	
	The number of nondeterministic bits required by $M_F$ is the same as $M_{F'}$ on input $(\sigma(x,k),h(x,k))$.
	Consequently, the computation of $M_F$ is still $k$-bounded as the number of nondeterministic bits is bounded by 
	\[
	f'(h(x,k))\cdot\log|\sigma(x,k)|\leq f''(k)\cdot\log|x|,
	\]
	where $f',f''$ are computable functions.
	This is due to $|\sigma(x,k)|$ being fpt-bounded and $h(x,k)$ being bounded by a computable function in $k$.
	
	We continue with $\sh\parab\L$.
	The proof is analogous but $M_{F'}$ has read-once access to its nondeterministic bits.
	Now, the only time $M_F$ accesses nondeterministic bits is when $M_{F'}$ accesses its nondeterministic bits.
	This simulation also preserves the order in which nondeterministic bits are accessed.
	So $M_F$ has read-once access to its nondeterministic bits as well.
\end{proof}

\begin{restatetheorem}[prop:closure-log]
\begin{theorem}
For any $o\in\{\W{},\W1,\beta,\beta\mbox{-}\complClFont{tail}\}$, the class $\sh\complClFont{para}_o\mbox-\L$ is closed under addition and multiplication.
\end{theorem}
\end{restatetheorem}
\begin{proof}
Let $F_1, F_2$ be in $C$ via the NTMs $M_1$ and $M_2$, respectively.
We start by showing that the above classes are closed under addition.
We first consider $\sh\paraw\L$ and $\sh\parab\L$.  The argument is similar for both of the classes. 
We give details for $\sh\paraw\L$. 
We construct a new machine $M$ as follows: $M$ nondeterministically chooses whether to simulate the machine $M_1$ or the machine $M_2$ on the given input using a single additional nondeterministic bit.
By construction we have $\sh\acc_M(x,k) = \sh\acc_{M_1}(x,k) + \sh\acc_{M_2}(x,k) = F_1(x,k) + F_2(x,k)$.
Also, $M$ is $O(\log|x| + h(k))$ space-bounded and $k$-bounded since it only behaves in the same way as either $M_1$ or $M_2$ and uses one additional (nondeterministic) step in the beginning.
We conclude that the class $\sh\paraw\L$ is closed under addition. 
Exactly the same argument works for $C=\parab\L$.

For $\sh\parawt\L$ and $\sh\parabt\L$, we modify $M$ as follows: $M$ first simulates the deterministic parts up until the first nondeterministic step of both machines $M_1$ and $M_2$. We then choose nondeterministically whether to finish the computation of $M_1$ or $M_2$.
This ensures that $M$ is tail-nondeterministic if $M_1$ and $M_2$ are tail-nondeterministic. Also, read-once access to nondeterministic bits is still preserved.

We now show that the above classes are closed under multiplication, starting again with $\sh\paraw\L$ and $\sh\parab\L$.
We construct a new machine $M'$. 
On input $(x,k)$, $M'$ first simulates $M_1$ on input $(x,k)$.
If $M_1$ accepts, then $M'$ simulates $M_2$ on $(x,k)$, and accepts if and only if $M_2$ does so. 
Since $M_1$ and $M_2$ are $k$-bounded, $M'$ is also $k$-bounded. The space used by $M'$ is at most the maximum of that of $M_1$ and $M_2$. 
By construction we have $\sh\acc_{M'}(x,k)= \sh\acc_{M_1}(x, k)\cdot\sh\acc_{M_2}(x,k) = F(x,k) \cdot F'(x, k)$. 
Accordingly, the class $C = \sh\paraw\L$ is closed under multiplication. 
Additionally, if $M_1$ and $M_2$ only have read-once access to the nondeterministic bits, so has the new machine $M'$. 
We conclude that $\sh\parab\L$ is closed under multiplication.

For $\sh\parawt\L$ and $\sh\parabt\L$, we modify $M'$ similar to what was done to show that these classes are closed under addition. 
$M'$ first simulates the deterministic parts up until the first nondeterministic step of both machines $M_1$ and $M_2$. 
Then $M'$ simulates $M_1$ starting from the first nondeterministic step. If $M_1$ accepts, $M'$ simulates $M_2$ starting from the first nondeterministic step and accepts if and only if $M_2$ accepts. 
Formally, let $C_1$ be the first nondeterministic configuration of $M_1$ on input $(x,k)$ and $C_2$ be the first nondeterministic configuration of $M_2$ on $(x,k)$. 
The new machine $M'$ runs $M_1$ on $(x,k)$ until the configuration $C_1$ is reached. 
Then $M'$ stores the configuration $C_1$ and runs $M_2$ on $(x,k)$ until the configuration $C_2$ is reached. 
Storing $C_2$ in a separate tape, $M'$ now proceeds with the nondeterministic computation of $M_1$ starting with the configuration $C_1$. 
For every accepting configuration of $M_1$, the machine $M'$ re-starts the computation of $M_2$ starting from configuration $C_2$ accepting on all paths where $M_2$ does. 
By construction, the machine $M'$ accepts $(x,k)$ if and only if $M_1$ and $M_2$ both accept it. 
As a result, we have 
  $\sh\acc_{M'}(x,k)= \sh\acc_{M_1}(x, k)\cdot\sh\acc_{M_2}(x,k) = F(x,k) \cdot F'(x, k).$
    
Suppose $M_1$ and $M_2$ are both tail-nondeterministic. 
Then the number of steps of $M_1$ (respectively $M_2$) starting from configuration  $C_1$ (respectively $C_2$) is $k$-bounded. 
The machine $M'$ does not perform any nondeterministic steps until it starts simulating $M_1$ from configuration $C_1$. 
By the tail-nondeterminism of $M_1$, the number of steps starting from $C_1$ until its termination is $k$-bounded. Once $M_1$ terminates, $M'$ starts simulation of $M_2$ starting from $C_2$ on all accepting computations of $M_1$. 
The computation of $M_2$ starting from the configuration $C_2$ is also $k$-bounded. 
This ensures that if both $M_1$ and $M_2$ are tail-nondeterministic, so is $M'$. 
Also, it can be seen that the read-once access to nondeterministic bits is preserved.
\end{proof}

\begin{restatetheorem}[thm:logreach-parabetaL-paralog-complete]
	\begin{theorem}
		For every $b\ge2$, $\p\sh\LOG\reach\bdegree$ is $\sh\paraBetaL$-complete \wrt $\leqparalogpars$-reductions.
	\end{theorem}
\end{restatetheorem}
\begin{proof}
 \begin{algorithm}[t]
 	\DontPrintSemicolon
 	\SetKwInOut{Input}{Input}
 	\caption{Algorithm solving $\p\sh\LOG\reach\bdegree$ in $\sh\paraBetaL$}\label{alg:reach}
 \Input{$\G = (V,E)$, $s,t \in V$, $a \in \mathbb{N}$, $k \in \mathbb{N}$}
\lIf{$a > k\cdot \log |V|$}{
 {\bfseries reject}\label{ln:checka}
 }
$v_1 \gets s$\;
 \For{$1 \leq i \leq a$\label{ln:forloop}}{
 guess a number $j$ between $1$ and $b$\;\label{ln:guess-a-number}
 \lIf{$v_1$ has less than $j$ successors}{
 {\bfseries reject}
 }
 let $v_2$ be the $j$th successor of $v_1$\;
 $v_1 \gets v_2$\;
 }
 \lIf{$v_1 = t$}{
 {\bfseries accept\ \ else\ \ reject}
 }
 \end{algorithm}
 Algorithm~\ref{alg:reach} shows membership. 
 The algorithm first checks the constraint on $a$.
 Then it starts from vertex $s$ and guesses an arbitrary path of length exactly $a$, using the fact that the out-degree of all vertices is bounded by $b$ to limit the number of nondeterministic bits needed for this task. 
 More precisely, we choose a natural ordering depending on the encoding of vertices and use this to reference successors of the current node by numbers $0,\dots,b-1$.
 The needed number of nondeterministic bits is $a \cdot \log(b) \in O(k \cdot \log |V|)$. 
 Furthermore, at any point of  time, indices of at most   two vertices as well as one number bounded by a constant are stored, so the algorithm uses logarithmic space.

 Regarding the lower bound, let $F \in \sh\paraBetaL$ via the machine $M$. 
 Using Lemma~\ref{normalise-lemma}, we can assume that $M$ has a unique accepting configuration $C_\acc$ and there is a computable function $f$ such that for any input $(x,k)$ all accepting computation paths of $M$ on input $(x,k)$ use exactly $f(k) \cdot \log|x|$ nondeterministic bits.
Let $(x,k)$ be an input of $M$. 
Consider the graph $\G(x,k) = (V(x,k),E(x,k))$ where $V(x,k)$ is the set containing all nondeterministic configurations of $M$ on input $(x,k)$ as well as $C_\acc$ and $E(x,k)$ contains an edge $(C,C')$ if and only if $C'$ is reached from $C$ using exactly one nondeterministic step and  any number of deterministic steps in the computation of $M$ on input $(x,k)$.
Now the number of accepting computation paths of $M$ on input $(x,k)$ is exactly the number of paths of length $f(k) \cdot \log|x|$ from the first nondeterministic configuration $s(x,k)$ reached in the computation of $M$ on input $(x,k)$ to the unique accepting configuration in $\G$.
To ensure that the length of accepting paths is bounded by $f(k)\cdot \log|\G|$, we further assume, \wLOG, that $|V(x,k)| \geq |x|$.
 
Note that no (directed) cycle is reachable from the initial configuration in the configuration graph of $M$ on any input. For that reason no cycle is reachable from $s(x,k)$ in $\G$ and hence every $s$-$t$-walk in $\G$ is an $s$-$t$-path.
We now have for all $(x,k)$ that
\begin{align*}
F(x,k) &= \sh\acc_M(x,k)\\
	   &= \p\sh\LOG\reach\bdegree((\G(x,k), s(x,k), t, f(k) \cdot \log|x|),f(k)).
\end{align*}
Adjacency within $\G(x,k)$ can be computed from $(x,k)$ in parameterised logspace, since it only depends on the (fixed) machine $M$ as well as computing deterministic paths. Furthermore, the new parameter is bounded by a computable function in the old parameter. 
Accordingly, the construction yields a para-logspace parsimonious reduction.
\end{proof}

\begin{restatetheorem}[thm:reach-parabetal-complete]
\begin{theorem}
	$\p\sh\reach$ is $\sh\parab\L$-complete \wrt $\leqparalogpars$.
\end{theorem}	
\end{restatetheorem}

\begin{proof}
Regarding membership, we modify Algorithm~\ref{alg:reach}.
Line~\ref{ln:checka} is omitted.
In line~\ref{ln:forloop}, we now loop over $i$ from $1$ to $k$ and in line~\ref{ln:guess-a-number} we guess a number between $1$ and $\log|V|$.
This number is used to choose a successor $v_2$ of $v_1$.

Regarding hardness, let $F \in \sh\parab\L$ via the machine $M$.
\WLOG, $M$ is in the normal form from Lemma~\ref{normalise-lemma} and there is a computable function $f$ \ST $M$ on any input $(x,k)$ uses exactly $f(k) \cdot \log|x|$ nondeterministic bits.
Also assume without loss of generality that the accepting configuration of $M$ is deterministic.
Fix an input $(x,k)$. 
We reduce the problem of counting the accepting computation paths of $M$ on input $(x,k)$ to the problem of counting walks in a modified version of the configuration graph of $M$ on input $(x,k)$. 
The difference to the hardness proof of Theorem~\ref{thm:logreach-parabetaL-paralog-complete} is that edges now encode computations comprised of $\log|x|$-many nondeterministic steps.
More precisely, define $\G = (V,E)$ and $s,t \in V$ as follows:
$V$ consists of all nondeterministic configurations of $M$ on input $(x,k)$ and the (unique) accepting configuration $C_\acc$ of $M$.
For $C' \neq C_\acc$, $(C,C') \in E$ if and only if configuration $C'$ is reachable from $C$ in exactly $\log|x|$-many nondeterministic steps (and any number of deterministic steps) in the computation of $M$ on input $(x,k)$. Furthermore, $(C, C_\acc) \in E$ if and only if $C_\acc$ is reachable from $C$  using only deterministic steps in the computation of $M$ on input $(x,k)$.
Finally, $s$ is the first nondeterministic configuration reached in the computation of $M$ on input $(x,k)$ and $t = C_\acc$.

Since all accepting computation paths in the configuration graph of $M$ on input $(x,k)$ use exactly $f(k) \cdot \log|x|$ nondeterministic bits, the change made compared to the reduction used for Theorem~\ref{thm:logreach-parabetaL-paralog-complete} does not change the number of paths and we can  simply  count paths of length exactly $f(k)$ in the new graph.
By the above we have
\begin{align*}
  F(x,k) &= \#\acc_M(x,k)\\
  &= \p\sh\reach((\G, s(x,k), t), f(k)).
\end{align*}
The set $E$ can still be computed by a $\para\L$-machine:
To check whether an edge $(C,C')$ is present, we simply loop over all values for the next $\log|x|$ nondeterministic bits and verify whether the corresponding sequence of configurations starting from $C$ is a path ending in $C'$ in the configuration graph of $M$ on input $(x,k)$.
Consequently, the construction yields a para-logspace parsimonious reduction.
\end{proof}

\begin{restatetheorem}[thm:reach-parabetataill-complete]
\begin{theorem}
  $\p\sh\reach$ is $\sh\parabt\L$-complete \wrt $\leqparalogpars$. 
\end{theorem}
\end{restatetheorem}

\begin{proof}
 This almost immediately follows from the proof of Theorem~\ref{thm:reach-parabetal-complete}.
 For membership, observe that the runtime of the modified algorithm described in that proof is bounded by $O(k \cdot \log n)$ and hence the algorithm is already tail-nondeterministic.
 For hardness, it suffices that the problem is $\sh\paraBetaL$-hard and $\sh\parabt\L \subseteq \sh\paraBetaL$.
 \end{proof}

\begin{restatetheorem}[thm:pathb-parabetal-complete-turingplog]
\begin{theorem}
	$\p\sh\LOG\WALK\bdegree$ is $\sh\paraBetaL$-complete \wrt $\leqparaturing$.
\end{theorem}	
\end{restatetheorem}
\begin{proof}
	For membership, we use Algorithm~\ref{alg:reach} but nondeterministically guess nodes $s,t\in V$.

	For hardness, let $F\in\p\sh\paraBetaL$ via the machine $M$.
  \WLOG assume that $M$ has a unique accepting configuration and that there is a computable function $f$  \ST $M$ on input $(x,k)$ uses exactly $f(k) \cdot \log|x|$ nondeterministic bits.
	
	Similarly as in the proof of Theorem~\ref{thm:logreach-parabetaL-paralog-complete}, let $\G(x,k)=(V(x,k),E(x,k))$ be the configuration graph such that all paths through only deterministic configurations are substituted by a single edge.
	Furthermore, we extend $\G(x,k)$ such that we add a path of fresh vertices $v_1,\dots,v_{\log|x|}$ with $(v_i,v_{i+1})\in E(x)$ for $1\leq i<\log|x|$.
	The reason for this construction lies in possible ``bad'' sequences of configurations as depicted in Figure~\ref{fig:badsequences}.
	\begin{figure}
	\centering
	\includegraphics[width=10cm]{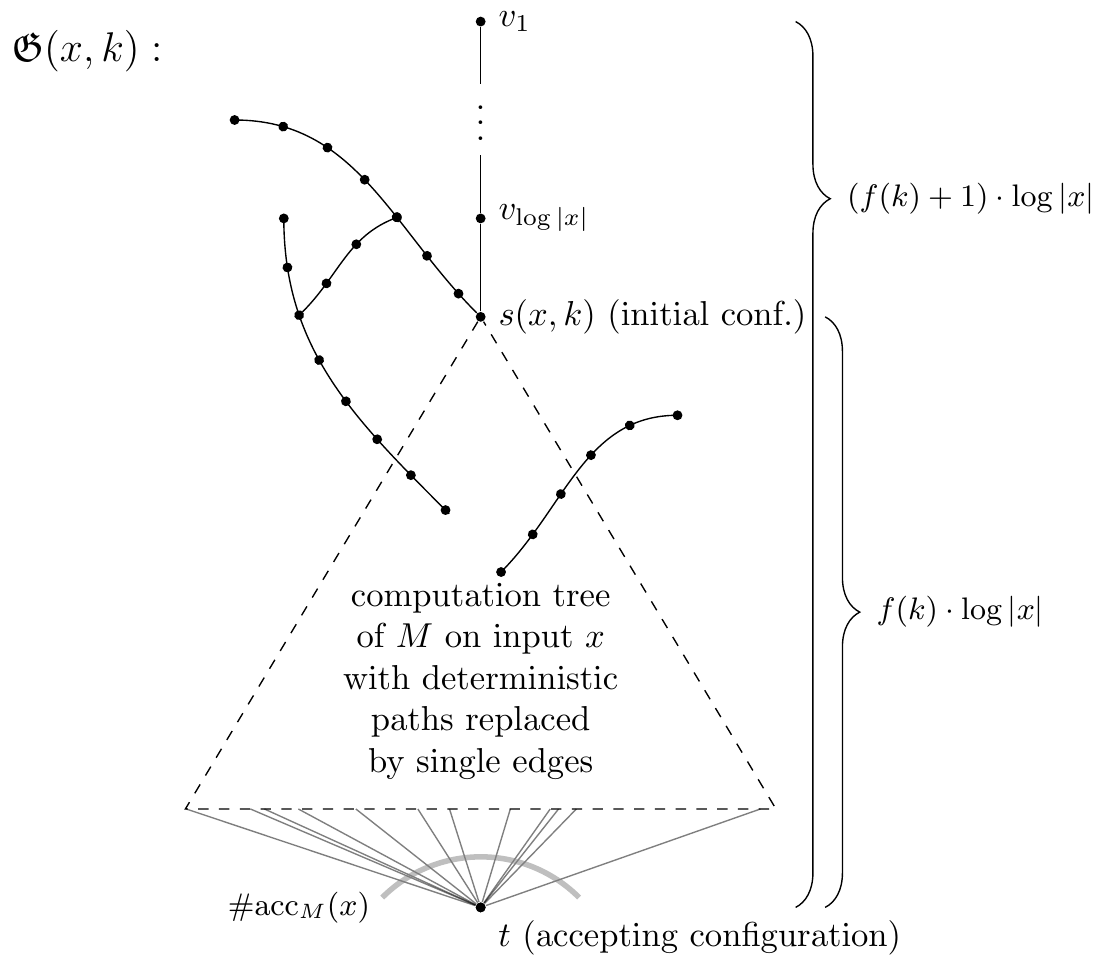} 
	\caption{Construction of $\G(x)$ in the proof of Theorem~\ref{thm:pathb-parabetal-complete-turingplog}. The black chains of unnamed vertices depict possibly occurring ``bad'' configuration sequences in the configuration graph that should not be considered as they are unreachable from the initial configuration.}\label{fig:badsequences}
	\end{figure}
  Furthermore, \wLOG we assume $|V(x,k)| \geq |x|$.
	Now, any path in $\G(x,k)$ going from $v_1$ to the initial configuration $s(x,k)$ and then to the accepting configuration $t$ is of length $\ell\dfn (f(k)+1)\cdot\log|x|$.
	Notice, that the number of such $v_1$-$t$-paths is equivalent to the number of accepting paths of $M$ on input $x$.
	To precisely calculate $\sh\acc_M(x,k)$, we will use two oracle queries to different $\p\sh\LOG\WALK\bdegree$-instances yielding a $\leqparaturing$-reduction as required. 

	At first, we compute the result $n_1$ of the oracle query $\p\sh\LOG\WALK\bdegree(\G(x,k),\ell,f(k)+1)$.
	Then, we modify $\G(x,k)$ yielding a graph $\G'(x,k)$ by deleting the edge $(v_1,v_2)$.
	This ensures that in $\G'(x,k)$ among paths of length $(f(k)+1)\cdot \log |x|$ exactly the ``good'' accepting paths are missing compared to $\G(x,k)$. 
	Then we store the value of the oracle query $\p\sh\LOG\WALK\bdegree(\G'(x,k),\ell,f(k)+1)$ in the variable $n_2$.
 	Finally, calculate the difference $n_1-n_2$ which is equivalent to $\sh\acc_M(x,k)$.
\end{proof}

\begin{restatetheorem}[thm:reach2cnf-parawl-complete]
\begin{theorem}
	$\p\sh\LOG\reach_2\CNF$ is $\sh\paraWL$-complete \wrt $\leqparalogpars$.
\end{theorem}	
\end{restatetheorem}

\begin{proof}
Regarding membership, we can modify Algorithm~\ref{alg:reach} to find paths as before. 
We then need to check whether the chosen path also satisfies the formula φ. 
For this, we do the following.
First, we use the folklore logspace-algorithm for propositional model-checking.
Then, whenever we need the value of a variable $e$, we re-compute the whole path constructed using Algorithm~\ref{alg:reach} reusing the nondeterministic bits and check whether the edge $e$ is used in that path (when $e$ occurs, then the computation can be stopped and the variable $e$ is true; when the computation terminates without $e$ occurring, then the variable $e$ is false). 
This yields a $\sh\paraWL$-algorithm.

Regarding hardness, we use the same construction as in the proof of Theorem~\ref{thm:logreach-parabetaL-paralog-complete}. 
The difference is that nondeterministic bits can be reused. 
This means that we need to ensure to only count paths on which different queries to the same nondeterministic bit assume the same value for that bit. 
Let $\G(x,k) = (V(x,k), E(x,k))$, $s(x,k)$, $t$ and $f$ be as in the proof of Theorem~\ref{thm:logreach-parabetaL-paralog-complete}. 
Let $\ell \colon E(x,k) \to \mathbb{N}$ be a labelling function stating which nondeterministic bit is read on each edge of $\G(x,k)$ and let $\textrm{val} \colon E(x,k) \to \{0,1\}$ be a function mapping each edge to the value of the corresponding nondeterministic bit assumed in the computation when using that edge. 
We can now define
\[φ(x,k) = \bigwedge_{\substack{e_1, e_2 \in E(x,k),\\\ell(e_1) = \ell(e_2),\\\textrm{val}(e_1) \neq \textrm{val}(e_2)}} \lnot e_1 ∨ \lnot e_2,\]
expressing that the values assumed for the nondeterministic bits are consistent throughout a path. 
Notice that $φ(x,k)$ can be computed in parameterised logspace.
Hence, the mapping
\[(x,k) \mapsto ((\G(x,k), s(x,k), t, φ(x,k), f(k) \cdot \log|x|), f(k))\]
is a para-logspace parsimonious reduction.
\end{proof}

\begin{restatetheorem}[thm:cyclecover-parawl]
\begin{theorem}
	$\p\sh\cyclecover_2\CNF$ is $\sh\paraWL$-complete \wrt $\leqparalogpars$.
\end{theorem}
\end{restatetheorem}

\begin{proof}
  The following algorithm shows membership.
  Let the input be the directed graph $\G = (V,E)$, the CNF-formula $φ$ and $a, k \in \mathbb{N}$.
  Guess distinct vertices $v_1, \dots, v_\ell \in V$ where $\ell \leq k$.
  These are supposed to be vertices from distinct cycles in the cycle cover.
  Now, for each $v_i$ guess a walk from that vertex to itself not containing any vertex whose encoding is lexicographically smaller than that of $v_i$.
  Reusing nondeterministic bits, it can be assured that no vertex is visited twice.
  In the above procedure, count the number of visited vertices.
  If that number is $k \cdot a$, then accept.
  This way, each accepting path corresponds to a distinct cycle cover of $\G$ with the desired properties.

  For hardness, we proceed similarly as for $\p\sh\LOG\reach_2\CNF$ in the proof of Theorem~\ref{thm:reach2cnf-parawl-complete}.
  Let $\G(x,k) = (V(x,k), E(x,k)), s(x,k), t$ and $f$ be defined as in the proof of that theorem.
  Now let $\G'(x,k)$ be the graph $\G(x,k)$ with added self-loops for all vertices except $s(x,k)$ and $t$, as well as the additional edge $(t,s(x,k))$.
  In $\G(x,k)$, the component reachable from $s(x,k)$ does not contain any cycles.
  Furthermore, $s(x,k)$ does not have a self-loop in $\G'(x,k)$.
  For these reasons, the edge $(t, s(x,k))$ is contained in any cycle cover of $\G'(x,k)$.
  The cycle containing this edge always consists of this edge as well as a path from $s(x,k)$ to $t$, that is, an accepting computation path.
  By construction, all accepting computation paths have length exactly $f(k) \cdot \log |x|$.
  For this reason, we can simply choose $k' = f(k)$ and $a = \log |x|$.
  Now cycle covers of $\G'(x,k)$ with $\leq f(k)$ non-self-loop-cycles and exactly $k\cdot a$ non-trivially covered vertices cover exactly one accepting computation path non-trivially and cover all other vertices by self-loops.
  Finally, $φ(x,k)$ is chosen as in the proof of Theorem~\ref{thm:reach2cnf-parawl-complete}.

  The desired reduction function is $(x,k) \mapsto ((\G'(x,k), φ(x,k), f(k) \cdot \log |x|), f(k))$.
  \end{proof}

\begin{restatetheorem}[thm:mc-local-bounded]
\begin{theorem}
	For $a\ge2,r\ge1$, $\p\sh\MC{\Sigmarlocal}_a$ is $\sh\paraBetaL$-complete and $\sh\parabt\L$-complete \wrt $\leqparalogpars$.
\end{theorem}	
\end{restatetheorem}

\begin{proof}
  We show both completeness-results by showing membership in $\sh\parabt\L$ and hardness for $\sh\parab\L$.

  For membership, let $φ \in {\Sigmarlocal}_a$ be a formula over some vocabulary σ, $\struc{A}$ be a σ-structure and $k \in \mathbb{N}$.
  \WLOG, $|φ| \leq k$, since we can immediately reject otherwise.
  When evaluating φ in $\struc{A}$ in a topdown-fashion, book-keeping only needs space bounded by a computable function in $k$, since $|φ| \leq k$.
  With the naive approach, we need more than logarithmic space to store the first-order assignment.
  To avoid this, we use the fact that φ is $r$-local and the arity of relation symbols occurring in φ is bounded by $a$.
  We nondeterministically guess the assignment for any first-order variable in φ when we first encounter an atom involving that variable and due to $r$-locality we can discard it after evaluating the next $r$ atoms.
  This means that at any point during the computation, we store an assignment only for the variables that occurred in the previous $r$ atoms, that is, for at most $a\cdot r$ variables.
  Note that we evaluate φ node by node and for each node only have to process a constant number of elements of $\dom(\struc A)$, which are of logarithmic length.
  Since $|φ| \leq k$, this means that the runtime of the whole procedure is bounded by $f(k) \cdot \log|\struc A|$ for some computable function $f$ and hence the procedure is tail-nondeterministic.

  To show hardness, we reduce from $\p\sh\reach$, which is $\sh\paraBetaL$-complete by Theorem~\ref{thm:reach-parabetal-complete}.
  Define the vocabulary $σ\dfn (E,s,t)$. 
  Fix an input $((\G = (\textsf{V},\textsf{E}), \textsf{s},\textsf{t}), k)$ of $\p\sh\reach$. 
  The formula
  \[φ_k(x_1, \dots, x_k) \dfn (x_1 = s) ∧ E(x_1, x_2) ∧ E(x_2, x_3) ∧ \dots ∧ E(x_{k-1}, x_k) ∧ x_k = t\]
  expresses that a tuple of vertices $(v_1, \dots, v_k)$ is an $s^{\struc{A}}$-$t^{\struc{A}}$-walk in the input structure $\struc{A}$.
  We define the σ-structure $\struc{G}$ as $(\textsf{V}, \textsf{E}, \textsf{s}, \textsf{t})$.
  By construction, the number of satisfying assignments of φ in structure $\struc{G}$ is exactly the number of $\textsf{s}$-$\textsf{t}$-paths of length $k$ in $\G$.
  Also, $|φ|$ is bounded by $k$.
  As a result, the mapping $((\G, \textsf{s}, \textsf{t}), k) \mapsto ((φ_k, \struc{G}), |φ|)$ is the desired reduction.
\end{proof}

\begin{restatetheorem}[thm:mc-for-w1]
\begin{theorem}
	$p\text-\sh\MC{\Sigma_0}$ is $\sh\parawt\L$-complete \wrt $\leqparalogpars$.
\end{theorem}	
\end{restatetheorem}
\begin{proof}
Regarding membership in $\sh\parawt\L$, let $\varphi(x_1, \dots, x_n) \in\Sigma_0$ be a σ-formula and ${\struc{A}}$ be a σ-structure. Furthermore, let $A$ be the universe of $\struc A$. 
The counting machine for $|\varphi({\struc{A}})|$ can be described as follows. 
Nondeterministically guess assignments for $x_1,\ldots, x_k$ from $A$ and verify if $\varphi$ is true under the guessed assignments.
The resulting machine is $k$-bounded, tail-nondeterministic (\cf the proof of Theorem \ref{thm:mc-local-bounded}) and uses space at most $O(|\varphi| + \log n)$, where $n = |{\struc{A}}|$. 

Regarding hardness, let $F\in\sh\parawt\L$ and $M$ be a $k$-bounded $O(\log|x|+g(k))$ space-bounded NTM, with access to $g(k)\cdot\log|x|$ nondeterministic bits.
\WLOG, assume that $M$ is in the normal form from Lemma~\ref{normalise-lemma}.

Let $σ = (V, R, C_0, C_\acc)$.
Fix an input $(x,k)$.
We define a $σ$-structure $\struc A$ from the configuration graph of $M$ on input $(x,k)$.
Let $V^{\struc A}$ be the set of all nondeterministic configurations of $M$ on input $(x,k)$ along with the accepting configuration, $C_0^{\struc A}$ be the first nondeterministic configuration of $M$ on input $(x,k)$ and $C_\acc^{\struc A}$ be the accepting configuration of $M$.
Let $R^{\struc A}$ be a relation with arity $g(k)+2$ \ST $R^{\struc A}(C_1,C_2, b_1,\ldots,b_{g(k)})$ is true if an only if $M$ reaches configuration $C_2$ from $C_1$ using exactly $g(k)\cdot\log|x|$ nondeterministic steps with $b_1 \circ \dots \circ b_{g(k)}$ as the content of the choice tape.
Here, $b_i \in \{0,1\}^{\log|x|}$ for all $i$. 
Given $C_1,C_2$ and numbers $b_1,\ldots , b_{g(k)} \in \{1,\ldots, |x|\}$, testing if $R^{\struc A}(C_1,C_2, b_1,\ldots,b_{g(k)})$ is true or not can be done in $\para\L$.
Now, let $\dom(\struc A) = \{1,\ldots, |x|\} \cup V$ and define the following $σ$-formula:
\[φ = \bigwedge_{i=1}^{g(k)} V(x_i) \wedge \bigwedge_{i=1}^{g(k)} \neg(V(z_i))\wedge (C_0 =x_1) \wedge (C_\acc = C_{g(k)}) \wedge \bigwedge_{i=1}^{g(k)-1} R(x_i, x_{i+1}, z_1,\ldots, z_{g(k)}). \] 
 Note that φ has $2\cdot g(k)$ free variables. 
 For any assignment to the free variables, φ is true if and only if the assignment to $z_1,\ldots, z_{g(k)}$ represents a nondeterministic choice that leads to acceptance and the assignment to $x_1,\ldots, x_{g(k)}$ represents the corresponding sequence of configurations. 
 In fact, given an assignment to $z_1,\ldots, z_{g(k)}$, if $M$ accepts along this nondeterministic path, then there is a unique assignment to the variables $x_1,\ldots, x_{g(k)}$ that satisfies the formula φ.
 
 Combining the above observations with the arguments used in the proofs of Theorems~\ref{thm:reach-parabetal-complete} and~\ref{thm:mc-local-bounded}, we conclude that $\sh\acc_M(x,k) = |φ({\struc{A}})|$. 
 Given $(x,k)$, the structure ${\struc{A}}$ and φ can be computed in $\para\L$ by similar arguments as before.
 This completes the proof.
\end{proof}

\begin{restatetheorem}[thm:phom(p^*)-pbetaL-complete]
\begin{theorem}
	$\pcountHom{\class{P}^*}$ is $\sh\paraBetaL$-complete \wrt $\leqparalogpars$.
\end{theorem}	
\end{restatetheorem}
\begin{proof}
\begin{figure}[t]
	\centering
	\includegraphics[width=.7\linewidth]{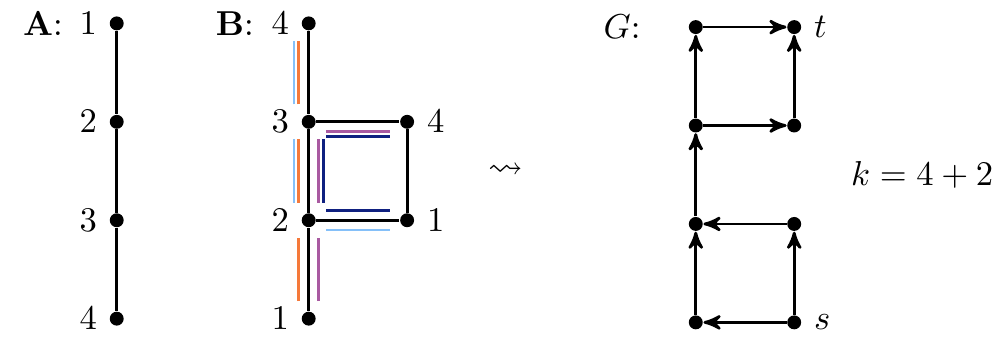}
	\caption{Example for the membership proof of Theorem~\ref{thm:phom(p^*)-pbetaL-complete}.}\label{fig:membership-phom-to-reachlog}
	
\end{figure}

	Regarding membership, we will show $\pcountHom{\class P^*}\leqparalogpars\p\sh\reach$, where the latter is in $\sh\paraBetaL$ by Theorem~\ref{thm:reach-parabetal-complete}.
	The proof idea is visualised in Figure~\ref{fig:membership-phom-to-reachlog} with an example. 
	Consider an arbitrary input $((\struc A,\struc B), k)$ to $\pcountHom{\class{P}^*}$ with \FO-structures $\struc A, \struc B$ and $k \in \mathbb{N}$. By definition of the problem we have $\struc A\in\class P^*$.
  Let $A = \dom(\struc A), B = \dom(\struc B)$.
  We assume that $\struc A$ and $\struc B$ have the same underlying vocabulary σ and $|A| \leq k$, since the remaining cases can easily be handled.
  Let $E$ be the edge relation symbol in σ.
	Furthermore, let $A=\{a_1,\dots,a_{n}\}$ with $n\in\N$ be the universe of $\struc A$ such that $(a_i,a_{i+1})\in E^{\struc A}$ and let $B$ be the universe of $\struc B$.
			
	Now, define the directed graph $\G=(B\cup\{s,t\},E')$ with $s,t\notin B$ and
	\begin{align*}
	E'\dfn& \{\,(s,x)\mid x\in C_{a_1}^{\struc B}\,\}\cup\\
  &\{\,(x,y)\in E^{\struc B}\mid \exists 1\leq i\leq n: x\in C_{a_i}^{\struc B}, y\in C_{a_{i+1}}^{\struc B} \,\}\cup\\
  &\{\,(x,t)\mid x\in C_{a_n}^{\struc B}\,\}.
	\end{align*}
  The reduction is given by the mapping $((\struc A, \struc B), k) \mapsto ((\G, s, t), |A|+2)$.
	
	We show correctness by giving a $1$-$1$-correspondence between homomorphisms from $\struc A$ to $\struc B$ and $s$-$t$-walks in $\G$ of length $|A|+2$.
	This correspondence is given by the following mapping.
	Let $h$ be a homomorphism from $\struc A$ to $\struc B$.
	This homomorphism is bijectively mapped to an $s$-$t$-path as follows.
	The homomorphism yields a sequence $(h(a_1),\dots,h(a_n))$ with $h(a_i)\in C_{a_i}^{\struc B}$, $(h(a_i),h(a_{i+1}))\in E^{\struc B}$ for $1\leq i<n$.
	From this, we obtain the $s$-$t$-path $(s,h(a_1),\dots,h(a_n),t)$ in $\G$. 
	
	Regarding injectivity, let $h\neq h'$ be two different homomorphisms from $\struc A$ to $\struc B$.
	As a result, there exists an $1\leq i< n$ such that $h(a_i)\neq h'(a_i)$.
	Consequently, $(s,h(a_1),\dots,h(a_n),t)\neq(s,h'(a_1),\dots,h'(a_n),t)$.
	
	Regarding surjectivity, let $(s,v_1,\dots,v_n,t)$ be an $s$-$t$-path in $\G$.
	Now, $(s,v_1)$, $(v_i,v_{i+1})$, $(v_n,t)\in E'$ with $1\leq i<n$.
	By construction of $E'$, this means $v_i\in C_{a_i}^{\struc B}$ for $1\leq i\leq n$ and $(v_i,v_{i+1})\in E^{\struc B}$ for $1\leq i<n$.
	Consequently, the function $h$ with $h(a_i) = v_i$ for all $i$ is a homomorphism from $\struc A$ to $\struc B$ and $h$ is a preimage of $(s, v_1, \dots, v_n, t)$ under the given mapping.
	
	The reduction can be computed by a $\para\Logspace$-machine as follows. 
	The vertex set of the graph is just the universe $B$ extended by two new vertices. 
	Store the sequence $a_1,\dots,a_n$.
	To compute $E'$, the machine identifies all vertices $x\in C_{a_1}^{\struc B}$ and prints for each the edge $(s,x)$; similarly for all edges $(x,t)$ with $x\in C_{a_n}^{\struc B}$.
	Then, for each edge $(x,y)\in E^{\struc B}$, we find and store the index $i$ such that $x\in C_{a_i}^{\struc B}$ and check whether $y\in C_{a_{i+1}}^{\struc B}$ is true.
	The machine only queries relations of $\struc B$ for individual tuples which can be achieved with binary counters.

	Regarding the lower bound, we first introduce a coloured variant $\p\sh\reachcolour$ of $\p\sh\reach$ as follows.
  Here, a colouring function $\ell\colon V\to \{1,\dots,m\}$ for some $m \in \mathbb{N}$ is given as an additional part of the input.
	\paracountingproblemdef{$\p\sh\reachcolour$}{Directed graph $\G=(V,E)$, $s,t\in V$, $\ell\colon V\to\{1,\dots,m\}$ with $m \in \mathbb{N}$, $\ell(s)=1$ and $\ell(t)=m$}{$k$}{number of $s$-$t$-paths $(s=v_1,\dots,v_k=t)$ with $\ell(v_i)=i$ for $1\leq i\leq k$, if $m = k$, 0 otherwise}
	The $\sh\parab\Logspace$-hardness with respect to $\leqparalogpars$ can be shown similarly to the one of $\p\sh\reach$ (as in the proof of Theorem~\ref{thm:reach-parabetal-complete}) as follows, see Lemma~\ref{lem:reach*} (the next result below) for details.
	The inputs $\G,s,t$ and $k$ can be constructed in the same way whereas the colouring function $\ell$ can be defined as follows.
	For each configuration $v\in V$, define $\ell(v)$ as the content of \tapeS, which is the value of the step counter.
	By definition of $\ell$, all $s$-$t$-paths respect the colouring.

	Now, we show $\p\sh\reachcolour\leqparalogpars\pcountHom{\class P^*}$.	
	Given a directed graph $\G=(V,E)$, $s,t\in V$, $k\in\N$, $\ell\colon V\to\{1,\dots,k\}$, we define two structures $\struc A,\struc B$ as follows.
	Let the universe of $\struc A$ be defined as $\{1,\dots,k\}$, $E^{\struc A}\dfn\{\,(i,i+1),(i+1,i)\mid 1\leq i<k\,\}$, and $C_{i}^{\struc A}\dfn\{i\}$ for $1\leq i\leq k$.
	Define the universe of $\struc B$ as $V$, $E^{\struc B}\dfn\{\,(u,v),(v,u)\mid (u,v)\in E\,\}$, and $C_i^{\struc B}\dfn \{\,u\mid\ell(u)=i\,\}$.
	
	Regarding correctness, the function that maps each $s$-$t$-path $\pi=(s=v_1,\dots,v_k=t)$ to the homomorphism $h_\pi$ with $h_\pi(i)=v_i$ for all $1\leq i\leq k$ is a bijective mapping between colour-respecting $s$-$t$-paths $\pi$ in $\G$ and homomorphisms $h_\pi$ from $\struc A$ to $\struc B$.
  For this, note that the colouring $\ell$ ensures that the direction of edges in the original graph is respected even for edges that are not part of actual computation paths on the current input. This is due to the fact that the step counter is added for any transition of the TM.
	The mapping is injective due to the fact that each node in the path contributes to the construction of the homomorphism.
	The mapping is surjective as for each homomorphism $h$ the path $(h(1),\dots,h(k))$ is a  colour-respecting $s$-$t$-path in $\G$ and a preimage of $h$.
	
	Regarding $\parab\Logspace$-computability, $\struc A$ relies only on the parameter.
	The universe of structure $\struc B$ is a copy, $E^{\struc B}$ is a symmetric closure of $E$, and the $C_i^{\struc B}$ require a $\log|V|$-counter to check for each $u\in V$ whether $\ell(u)=i$.
	This completes the proof.
\end{proof}

\begin{lemma}
\label{lem:reach*}
$\p\sh\reachcolour$ is $\sh\parab\L$-complete \wrt $\leqparalogpars$.
\end{lemma}
\begin{proof}[Sketch]
The proof is very similar to that of Theorem~\ref{thm:reach-parabetal-complete}. Regarding hardness, we outline the main difference and the process of labeling the vertices. Let $F\in \sh\parab\L$ via a $k$-bounded $\log$-space NTM $M$ which has read-once access to nondeterministic bits. Consider an input $(x,k)$ and  let $\G =(V,E)$ be the configuration graph of $M$ as defined in the proof of Theorem~\ref{thm:reach-parabetal-complete}. For a nondeterministic configuration $C\in V$, let $\ell(C)$ be the number represented in tape S in the configuration $C$. Now it is not hard to see that
 $F(x,k) = \p\sh\reachcolour(\G,s,t,\ell)$, where $s$ is the first nondeterministic configuration reachable from the initial configuration of $M$ on $(x,k)$ and $t$ is the unique accepting configuration of $M$.  
\end{proof}

\begin{restatelemma}[lem:clow]
\begin{lemma}
For all $(0,1)$-matrices $A$ and $k \in \mathbb{N}$ it holds that
\[\pdet(A,k) = \sum_{W \in {\cal W}_{\G_A, k}}\sign{W}\cdot\wt{W}.\]
\end{lemma}      	
\end{restatelemma}

 \begin{proof}
 The statement essentially follows from the arguments of Mahajan and Vinay~\cite[Theorem 1]{DBLP:journals/cjtcs/MahajanV97} and reader is referred to the original article for a full construction. 
 Their proof involves defining an involution $\eta$ (\ie, $\eta$ is a bijection whose inverse is itself) on the set of clow sequences such that $\eta$ is the identity on the set of all cycle covers (\ie, $\eta(C) = C$ for any cycle cover $C$) and for any clow sequence $W$ that is not a cycle cover, we have $\sign{W} = -\sign{\eta(W)}$.  

 We briefly describe the involution $\eta$ given by Mahajan and Vinay~\cite{DBLP:journals/cjtcs/MahajanV97}. 
 For a clow sequence ${W} =(W_1,\ldots,W_r)$, the clow sequence $\eta({ W})$ is obtained from $W$ as follows.
 Let $i \in \{1, \dots, r\}$ be the smallest index such that clows $W_{i+1}, \ldots, W_r$ are vertex  disjoint simple cycles.
 Traverse the clow $W_i$ starting from the head until we reach some vertex $v$ such that either $v$ is in some $W_j$ for $i+1 \le j\le r$ or $v$ completes a simple cycle within $W_i$.
 In the former case, we merge the simple cycle $W_j$ with the clow $W_i$ 
 and remove $W_j$ from the sequence to get the clow sequence $\eta(W)$.
 In the latter case, we split the clow $W_i$ at the vertex $v$ to get a new clow $W_i'$ and a simple cycle $C'$.
 Then $\eta(W)$ is the clow sequence obtained by replacing $W_i$ by the new clow $W_i'$ and inserting the cycle $C'$ into the resulting clow sequence.
 The remaining clows in $W$ are kept untouched.
  
 We note that the involution $\eta$ described above does not require that the clow sequence ${ W}$ has exactly $n$ edges.
 In particular, $\eta(W)$ is well defined even when $W \in {\cal W}_{\G,k}$ for some graph $\G$. 
 Furthermore, for $W \in {\cal W}_{\G,k}$ we have $\eta(W) \in {\cal W}_{\G,k}$ and $\sign{W} = -\sign{\eta(W)}$, since $\eta(W)$ has the same number of edges as $W$ and either has one clow more than $W$ or one clow less than $W$.
 Combining with the argument that $\eta$ is indeed an involution~\cite[Theorem 1]{DBLP:journals/cjtcs/MahajanV97}, we conclude:
    \[
    \pdet(A,k) = \sum_{W \in {\cal W}_{\G_A, k}}\sign{W}\cdot\wt{W}.\qedhere
    \]
\end{proof} 

\begin{restatetheorem}[thm:pdet-ub]
\begin{theorem}
The problem $\pdet$ for $(0,1)$-matrices can be written as a difference of two functions in $\sh\parabt\L$, and is $\sh\parabt\L$-hard \wrt $\leqparalogmet$.
\end{theorem} 	
\end{restatetheorem}
\begin{proof}
We prove this using a parameterised version of the algorithm given by Mahajan and Vinay~\cite[Thm.~2]{DBLP:journals/cjtcs/MahajanV97}. 
Let $A$ be the adjacency matrix of a directed graph $\G$.
We construct two $k$-bounded nondeterministic $\para\L$-machines $M_1$ and $M_2$ that have read-once access to their nondeterministic bits such that $\pdet(A,k) = \sh\acc_{M_1}(A,k) -\sh\acc_{M_2}(A,k).$

Both $M_1$ and $M_2$ behave exactly the same except for the last step, where their answer is flipped. 
More precisely, both machines nondeterministically guess a $k$-clow sequence, and $M_1$ accepts if  the guessed clow sequence has a positive sign while $M_2$ accepts if the guessed clow sequence has a negative sign. 
Then, it follows from Lemma~\ref{lem:clow} that $\pdet(A,k) =\sh\acc_{M_1}(A,k) -\sh\acc_{M_2}(A,k).$

Now, we describe the process of guessing nondeterministic bits. 
The process is the same for both $M_1$ and $M_2$. 
We need the following variables throughout the process: {\tt curr-head}, {\tt curr-vertex}, {\tt parity}, {\tt count}, {\tt ccount}.
The variable {\tt curr-head} contains the head of the clow currently being constructed, while {\tt parity} holds the parity of the partial clow constructed so far and is initialised to $(-1)^{2n-k}$.
Note that this initiation of {\tt parity} is because we are going to compute sign of a clow sequence $W$ as $(-1)^{2n-k+r'}$ where $r'$ is the number of clows in $W$ with at least two edges.
The variable {\tt count} keeps track of the total number of edges used in the partial clow sequence  constructed so far and {\tt ccount} keeps track of the number of edges in the current $k$-clow. 
The machines $M_1$ and $M_2$ are described in Algorithm~\ref{alg:det}.

Note that the guess $a=0$ in step 5 leads to expansion of the current clow with addition of a newly guessed edge.
The guess $a=1$ in step 5 leads to completion of the current clow by chosing the back edge to the head (step 10, only existence of such edge needs to be checked), and guessing head for a new clow.
Since the parity of the number of clows changes for the case when $a=1$ we flip the parity (step 14).
Note that the algorithm reaches step 17 if only if the nondeterministic choices correspond to a clow sequence with exactly $k$ edges.
Hence, for $b \in \{1,2\}$ the machine $M_b$ accepts on all nondeterministic paths where the guessed $k$-clow sequence has parity $(-1)^{b+1}$ which completes the correctness proof. 

Since both $M_1$ and $M_2$ guess exactly $k$ vertices, they are $k$-bounded. 
Also, only \texttt{curr-vertex} and \texttt{curr-head} need to be stored at any point of time, consequently the machines need only read-once access to nondeterministic bits.
Finally, the machines use $O(\log|A| + \log k)$ space and are tail-deterministic, taking only $O(k\cdot \log|A|)$ steps after the first nondeterministic step.
\begin{algorithm}[t]
 	\DontPrintSemicolon
 	\SetKwInOut{Input}{Input}
 	\caption{Machine $M_b$, $b\in \{1,2\}$.}\label{alg:det}
 \Input{$G = (V,E)$ as the adjacency matrix $A$. }
 Guess a vertex $v \in \{1,\ldots, n\}$\;
  $\texttt{curr-head}\gets v, \texttt{curr-vertex}\gets v$ \;
 $\texttt{parity}\gets (-1)^{2n-k}$, $\texttt{count}\gets 0$, $\texttt{ccount}\gets0$. \;
 \While{\normalfont$\texttt{count}\le k-1 $}{
	 Guess $a \in \{0,1\}$ \;
	 \eIf{$a=0$}{
		 Guess $v\in \{1,\ldots, n\}$  such that $(\texttt{curr-vertex},v) \in E$\;
		 $\texttt{curr-vertex}\gets v$, $\texttt{count}\gets \texttt{count}+1$, $\texttt{ccount}\gets \texttt{ccount}+1$\;
	 }{
		\lIf{\normalfont$\texttt{ccount}<1$ \textbf{or} $(\texttt{curr-vertex}, \texttt{curr-head}) \notin E$}{\textbf{reject}} 
		$\texttt{count}\gets \texttt{count}+1$\; 
		Guess $v\in\{1,\ldots, n\}$ such that $v>\texttt{curr-head}$\;
		$\texttt{curr-head}\gets v$ and $\texttt{curr-vertex}\gets v$\;
		$\texttt{parity}\gets -1 \cdot \texttt{parity}$\;
		$\texttt{ccount}\gets 0$
	 }
 }
 \lIf{\normalfont$(\texttt{curr-vertex}, \texttt{curr-head}) \notin E$}{\textbf{reject}}
 \leIf{\normalfont$(-1)^{b+1} = \texttt{parity}$}{\textbf{accept}}{\textbf{reject}}
  \end{algorithm}

For hardness we give a reduction from $\p\sh\reach$. 
Let $\G,s,t$ be an instance of $\p\sh\reach$, where $\G$ is a DAG. 
Let $\G'$ be the graph obtained by adding the ``back edge'' $(t,s)$ to $\G$. 
Note that the set of all $s$-$t$ paths in $\G$ is in bijective correspondence with the set of cycles in $\G'$. 
Let $A'$ be the adjacency matrix of $\G'$. 
Then $\pdet(A',k) = (-1)^{2n-k+1}\cdot R$ where $R$ is the number of $s$-$t$ paths in $\G$.  
As a result, $R$ can be retrieved from $\pdet(A',k)$ in deterministic logspace.
This completes the proof. 
\end{proof}

\section{Branching Programs}
\subsection*{Preliminaries}
A BP $P$ is a layered directed acyclic graph (DAG) with a \emph{source node} $s$ and a \emph{sink node} $t$ (see, e.g., textbook by Vollmer~\cite{DBLP:books/daglib/0097931}). 
The vertices of the BP are labelled by input variables in $\{x_1,\ldots, x_n\}$ and edges are labelled by $0$ or $1$. 
An input $a_1\cdots a_n\in \{0,1\}^n$ is \emph{accepted by} $P$ if there is a directed $s$ to $t$ path (short: $s$-$t$-path) $\rho$ that is consistent with $a$, that is, for each edge $(u,v)$ in $\rho$,
${\lab} (u,v) = a_i$ where ${ \lab}(u) = x_i$. 
The BP $P$ is said to be \emph{deterministic} (DBP for short) if every vertex has either no outgoing edge or exactly two outgoing edges, one labelled by $0$ and the other by $1$. 
The \emph{size} of the program $P$ is the number of vertices in it, the \emph{length} is the length of a longest path starting from $s$. 
If $P$ has length $\ell$, then we assume that the vertices of $P$ are partitioned into layers $L_0 \cup L_1\cup \ldots \cup L_\ell$ where $L_0$ contains the source and $L_\ell$ contains the sink. 
By layer $i$, we mean the set of vertices in $L_i$.

A \emph{family of branching programs} is a family $\mathcal{P} = (P_n)_{n \in \mathbb{N}}$ containing one BP $P_n$ for each input length $n \in \mathbb{N}$.
In this article, every family of BPs is implicitly logspace-uniform (in the context of classical complexity) or $\para\L$-uniform (in the context of parameterised complexity).
Details about notions of uniformity in the classical context can be found in the textbook of Vollmer~\cite{DBLP:books/daglib/0097931}.
We need a straightforward adaptation of these notions to the parameterised setting.

We let $\bp$ be the set of all languages accepted by polynomial-size families of BPs, and $\DBP$ be the set of all languages accepted by polynomial-size families of DBPs.
%
Now, we will incorporate the notion of parameters to BPs. 
For that reason, the families of BPs are of the form $\mathcal{P}\dfn(P_{n,m})_{n,m\ge0}$.
The language \emph{accepted by $\mathcal{P}$} is the set of all inputs $(x,k)\in\{0,1\}^*\times\N$ such that $P_{|x|,|k|}$ accepts $(x,k)$.
Let $\calc$ be a complexity class based on families of BPs of size $s(n)$.
A \emph{family of $\para\calc$-BPs} is a family of BPs $(P_{n,m})_{n,m\ge0}$ of size $s(n+f(m))$.

Note that, while the operator $\para$ is defined for arbitrary complexity classes, the operators $\paraw$, and $\parab$ are only defined with respect to TM based classes. 
These operators can be generalised to be also applicable to complexity classes based on BPs by extending the notion of $k$-bounded nondeterminism to this context, though.
A \emph{family of branching programs with nondeterministic input} has additional input bits that are marked as nondeterministic, similar to a choice tape for TMs.
Let $\mathcal{ P}\dfn (P_{n,m})_{n,m\ge0}$ be such a family and $\ell(n,m)$ be the number of nondeterministic input bits in $P_{n,m}$.
Then \emph{$\mathcal{P}$ accepts an input $(x,k)\in\{0,1\}^*\times\N$}, if there is a $y\in\{0,1\}^{\ell(|x|,|k|)}$ such that $P_{|x|,|k|}$ accepts $((x,k),y)$. 
Also, denote by $\sh\acc_P(x,k)$ the number of $y\in\{0,1\}^{\ell(|x|,|k|)}$ such that $\mathcal{P}$ accepts $((x,k),y)$.
While it might seem counter-intuitive at first, we also use the notion of \emph{deterministic} families of BPs with nondeterministic input, meaning that the program is deterministic if the nondeterministic input is fixed. 
Furthermore, $\mathcal{P}$ is \emph{$k$-bounded} if there exists a computable function $f$ \stfa $n,m\ge0$, the number $\ell(n,m)\le f(m)\cdot \log n$.


We will now introduce a notion of read-once access to nondeterministic bits for the above classes.
This notion is inspired from the notion of read-once certificates considered in~\cite{DBLP:books/daglib/0023084} and first defined in~\cite{DBLP:conf/fct/MahajanR09}.
Let $P(x,y)$ be a BP with two inputs $x=x_1\cdots x_n$ and $y=y_1\cdots y_m$. Here, $y$ is the nondeterministic input.
We say that $P$ is \emph{read-once certified} if there are layers $i_0< i_1 < i_2< \dots < i_m$ in the underlying graph of $P$ such that the variable $y_j$ occurs as a label only in layer numbers $\eta$ such that $i_{j-1}\le \eta\le i_j$.

\begin{definition}
Let $\calc$ be the class of languages accepted by families of BPs of size $O(s(n))$ for some $s\colon\N\to\N$.
Then, $\paraw\calc$ is the class of all parameterised languages computable by $k$-bounded families of $\para\calc$-BPs.
\end{definition}
\begin{definition}
Let $\calc$ be the class of languages accepted by families of BPs of size $O(s(n))$ for some $s\colon\N\to\N$.
Then, $\parab\calc$ is the class of all parameterised languages computable by $k$-bounded families of $\para\calc$-BPs that are read-once certified.
\end{definition}
Analogously, we also define the tail-nondeterministic variants of the operators $\paraw$ and $\parab$, called $\parawt$ and $\parabt$, respectively.

\subsection*{Parameterised Counting and Branching Programs}
Using the notion of bounded nondeterminism for the case of BPs, we define counting classes based on $\DBP$, which can be generalised to further classes of BPs.
\begin{definition} Let $F$ be a parameterised function.
Then, $F\in\sh\paraw\DBP$ if there exists a $k$-bounded family $\mathcal{P}=(P_{n,m})_{n,m\ge0}$ of DBPs of size $O(f(k)\cdot p(|x|))$ \stfa $(x,k)$: $F(x,k) = \sh\acc_{\mathcal{P}}(x,k)$.
We say that $F$ is also in 
\begin{itemize}
	\item $\sh\parab\DBP$ if there is such a $\mathcal{P}$ that is read-once certified,
	\item $\sh\parawt\DBP$ if there is such a $\mathcal{P}$ that is tail-nondeterministic, and
	\item $\sh\parabt\DBP$ if there is such a $\mathcal{P}$ that is read-once certified and tail-nondeterministic.
\end{itemize}
\end{definition}

$\NL$ coincides with the class of all languages accepted by logspace-uniform families of BPs of polynomial size~\cite{DBLP:conf/mfcs/Meinel86}. 
The desired family of BPs is obtained from the configuration graph of a nondeterministic logspace-bounded machine  preserving the number of accepting paths, and hence the result generalises to the corresponding counting classes.
We extend these relationships to the parameterised setting. 
We need the following construction based on configuration graphs.
\begin{lemma}
\label{lem:layering}
Let $M$ be a $k$-bounded nondeterministic \para\L-machine. Then  for any input $(x,k)$ there is a layered DAG  $\G_{M,n,m}$ with two special vertices $s$ and $t$ such that $\sh\acc_M(x,k)$ is the number of paths from $s$ to $t$ in $\G_{M,n,m}$. Given $M$ and $(x,k)$, the graph $\G_{M,n,m}$ can be computed in $\para\L$. 
\end{lemma}
\begin{proof}
In general, the configuration graph of a nondeterministic space-bounded machine is not layered and layering an arbitrary DAG might require reachability which is not possible in logspace-uniform fashion. 
However, we can modify the machine by adding a step counter, resulting in a layered configuration graph. We explain the details of this modification below.

 Let $M$ be a $k$-bounded, $s$ space-bounded TM.
 \WLOG, $M$ has a unique accepting configuration. 
 We now construct the machine $\widehat{M}$ modifying $M$ as follows: 
 In addition to the tapes of $M$, $\widehat{M}$ has an extra tape that keeps a binary step counter for $M$. 
 The machine $\widehat{M}$ simulates $M$ step by step and after each step of $M$, it increments the counter by one.
 By construction, $\widehat{M}$ behaves exactly as $M$ apart from the additional counter.
 Note that at most $O(s)$ bits are required to store the counter because the runtime of $M$ is bounded by $2^{O(s)}$.
 For a given input $x$, a configuration of $\widehat{M}$ on $x$ can be represented as $(\gamma,i)$ where $\gamma$ is a configuration of $M$ on the input $x$, and $i$ is the content of the additional tape of $\widehat{M}$ that contains the binary counter. 
 To simplify things, we assume that $i$ is a number in $\{0, \ldots, 2^{O(s)}\}$, ignoring the head position information on the tape used for storing the counter. The normalised configuration graph of $\widehat{M}$ is a graph on the set of all possible configurations of $\widehat{M}$ on any input of length $n$, with the edge relations as follows. 
 There is an edge from configuration $(\gamma ,i)$ to $(\gamma', i')$ if and only if $\gamma'$ is reachable from $\gamma$ in a single step of $M$ and $i' = i+1$. 
 That is, the steps required to update the counter are merged into the preceding step of $M$ in the normalised configuration graph. 

 The generic configuration graph $\G_{M,n,m}$ of $\widehat{M}$ on inputs $(x,k)$ with $x=x_1\cdots x_n$ and $m=|k|$ is defined as follows: The set of vertices is the set of all configurations $(\gamma,i)$ of $\widehat{M}$ on any input of length $n,m$ without information on the content of the input tape. Additionally, vertices are labelled with variables $\{x_1, \dots, x_n, k_1, \dots, k_m\}$ and edges are labelled from $\{0,1\}$ in the same way as for BPs: Vertex labels state which input bit is read in the next step of the computation while edge labels state for which value of the respective input bit that edge is used.
 It is not hard to see that the underlying graph of any generic configuration graph for input lengths $n,m\ge 0$ is a layered DAG.
 Moreover, it can be seen that for any $(x,k)\in \{0,1\}^*\times \mathbb{N}$, $M$ accepts $(x,k)$ if and only if there is a directed path from the initial configuration to the accepting configuration in $G_{M,|x|,|k|}$ such that the edge labels along the path are consistent with the bits of $x$ and $k$. 
 Finally, we may note that there is a $\para\L$-machine $N$ that, given $(C,i)$ and $(C',i')$, decides if $G_M$ has an edge from $(C,i)$ to $(C',i')$. 
 Accordingly, the construction of $\G_{M,n,m}$ is $\para\L$-uniform. 
\end{proof}

\begin{theorem}
\label{lem:equivalence-BP-SC}
For any $o\in\{\W{},\W1,\beta,\beta\mbox{-}\complClFont{tail}\}$, we have that $\sh\complClFont{para}_o\mbox-\DBP = \sh\complClFont{para}_o\mbox-\L$.
\end{theorem}
\begin{proof}
	We will first show that $\sh\complClFont{para}_o\mbox{-}\DBP \subseteq \sh\complClFont{para}_o\mbox{-}\L$ for $o$ as above, giving the detailed argument for $o = W$ and outlining the changes required (if any) for the remaining cases.
	
	Let $F \in \sh\paraw\DBP$ via the $\para\L$-uniform family $\mathcal{P}\dfn (P_{n,m})_{n,m\ge0}$ of BPs of size $g(m)\cdot n^c$ for $c\in\N$.
	Let $M$ be a $k$-bounded NTM $M$ which on input $(x,k)$ evaluates $P_{|x|,k}$ on that input, such that every $s$-$t$-path in $P_{|x|,k}$ corresponds to a unique accepting path of $M$.
	Algorithm~\ref{simulation-algorithm} describes the behaviour of a machine with these properties. Since $\mathcal{P}$
	is $\para\L$-uniform, there is a $\para\L$-machine $M'$ that on input $(1^n,k)$ outputs $P_{n,k}$ as a graph. Since $M$ cannot store $P_{|x|,k}$, as usual with composition of space-bounded algorithms, whenever $M$ needs to access an edge of $P_{|x|,k}$, $M$ runs the machine $M'$ from the start until the required edge is output or $M'$ halts, in which case the required edge is treated as absent in $P_{|x|,k}$. 
	\begin{algorithm}[t]
	\DontPrintSemicolon
	\SetKwInOut{Input}{Input}
	\caption{Algorithm that evaluates the BP in the proof of Theorem~\ref{lem:equivalence-BP-SC}.}\label{simulation-algorithm}
	
	\Input{$x=x_1\cdots x_n\in\{0,1\}^*$, $k\in\N$}
	Let $v_{\textit{cur}}$ be the starting vertex of $P_{|x|,k}$ and $t\leftarrow 0$\;
	\Repeat{$v_{\textit{cur}}$ is not the accepting node of $P_{|x|,k}$ and $t \leq g(k)\cdot |x|^c$}{
		\uIf{$v_{\textit{cur}}$ is a deterministic node in $P_{|x|,k}$}{
			move to next node in $P_{|x|,k}$ based on value of label $\ell\leftarrow x_i$\;
		}
		\ElseIf{$v_{\textit{cur}}$ is a nondeterministic node in $P_{|x|,k}$ and $label(v_{\textit{cur}}) = y_i$}{
			Guess a value $y_i=b$, nondeterministically and set $\ell\leftarrow b$\;
		}
		\lIf{edge ($v_{\textit{cur}}\xrightarrow{\ell}v_{\textit{next}}$) is in $P_{|x|,k}$}{
			$v_{\textit{cur}}\leftarrow v_{\textit{next}}$ and $t\leftarrow t+1$ \textbf{\ else\ \ reject}
		}
	}
	\lIf{$v_{curr}$ is the accepting node of $P_{|x|,k}$}{
		{\bfseries accept\ \ else\ \ reject}
	} 
	
	\end{algorithm}
	
	 By construction $\sh\acc_M(x,k)=\sh\acc_P(x,k)$ for all $(x,k)\in\{0,1\}^*\times\N$.
	 Furthermore, $M$ is $h(k)\cdot |x|^{O(1)}$ space-bounded for some computable $h$, since it only needs to store two vertices of $P_{|x|,k}$ as well as the space needed to construct $P_{|x|,k}$.
	 The machine is $k$-bounded, because $\mathcal{P}$ is $k$-bounded.
	 This concludes the proof for $o = W$.
	 If $P_{n,m}$ is read-once certified for all $n,m$, then $M$ requires only a read-once access to nondeterministic bits. 
	 If $P_{n,m}$ is tail-nondeterministic, so is $M$. 
	
	 Now, we prove $\sh\complClFont{para}_o\mbox{-}\DBP \supseteq \sh\complClFont{para}_o\mbox{-}\L$ for $o$ as in the statement of the theorem.
	 Our argument crucially uses the fact that the generic configuration graph $\G_{M,n,m}$ of a $k$-bounded machine is layered and can be constructed in $\para\L$ on input $(x,k)$ with $|x|=n, |k|=m$ (Lemma~\ref{lem:layering}). 
	 As in the case of $``\subseteq$'', we argue for the case of $\sh\paraw\L$ and mention the modifications required (if any) for the remaining classes. 
	 Let $F\in\sh\paraw\L$ via the $k$-bounded machine $M$ using $O(\log n + g(k))$ space on all inputs $(x,k)$.
	 \WLOG, we assume that the machine $M$ reads either from the input tape or from the choice tape in any configuration. 
	 Let $P_{n,m}\dfn \G_{M,n,m}$ be the generic configuration graph of $M$ for input length $n$ and $m=|k|$. 
	 Then, $M$ accepts $(x,k)$ if and only if there is an assignment $y\dfn y_1,\ldots, y_\ell \in \{0,1\}$ to the nondeterministic input such that $P_{n,m}$ has a directed path consistent with the input $((x,k),y)$ from the initial configuration to the accepting configuration. 
	 In fact, there is a $1$-to-$1$-correspondence between accepting computation paths of $M$ and choices for $y$.
	 As a result, $\sh\acc_P(x,k) = \sh\acc_M(x,k)$ as required. 
	 Also, $P_{n,m} = \G_{M,n,m}$ can be constructed using a $\para\L$ uniformity machine. 
	 This shows that $\sh\paraw\L \subseteq \sh\paraw\DBP$. 
	 
	 For the case of $\parab\L$, we need to show that the resulting DBP $P_{n,m}$ is read-once certified. 
	 However, it may happen that the configuration $(C,i)$ reads variable $y_j$, whereas $(C', i)$ reads variable $y_{j'}$ with $j\neq j'$. 
	 This makes $P_{n,m}$ far from being read-once certified. 
	 However, a crucial observation is that in any start-to-terminal path in $P_{n,m}$, the $y$-variables are read in the order $y_1,\ldots, y_\ell$, and if $y_j$ is read at any point, then none of the $y_i$'s for $i<j$ will be read along this path after this point. 
	  Accordingly, with suitable staggering we can make $P_{n,m}$ read-once certified. 
	 We sketch the process below.
	 
	 For $1\le i\le m$, let $i_\alpha$ be the largest number such that there is a configuration $C$ such that $(C,i_\alpha)$ reads the nondeterministic bit $y_i$. 
	 We delay the computations so that the last read of $y_i$ occurs until layer $i_\alpha$ and the variable $y_{i+1}$ is read only after layer $i_\alpha$.
	 That is, for $y_1$, we wait till layer number $1_\alpha$ before proceeding to read $y_2$, and for $y_2$ we wait till layer $2_\alpha$ before proceeding to read $y_3$ and so on. 
	 This can be achieved by adding necessary dummy nodes that have a single outgoing edge labelled by $1$. 
	 For the whole process we only need the value of $i_\alpha$ which can be computed in $\para\L$ given access to the uniformity machine. 
	 Consequently, the overall staggering process can be done in $\para\L$ so that the resulting BP is uniform. 
	 It may be noted that we do not alter the number of accepting paths during the above process.
	 
	 From the above, we conclude that $\sh\parab\L \subseteq \sh\parab\DBP$. 
	 Finally, in the case of tail-nondeterminism, we may note that if $M$ is tail-nondeterministic then so is $\G_{M,n,m}$ for every input length $n$. 
	 Since the above staggering process does not alter tail-nondeterminism, we conclude that $\sh\parawt\L \subseteq \sh\parawt\DBP$ and $\sh\parabt\L \subseteq \sh\parabt\DBP$.
	\end{proof}

\begin{theorem}\label{thm:bp-nl-equality}
For any $o\in\{\W{},\W1,\beta,\beta\mbox{-}\complClFont{tail}\}$, we have that $\sh\complClFont{para}_o\mbox-\BP = \sh\complClFont{para}_o\mbox-\NL$.
\end{theorem}
\begin{proof}[Proof Idea.]
	In the proof of Theorem~\ref{lem:equivalence-BP-SC}, for the graph $P_{n,m}= \G_{M,n,m}$, the properties such as $k$-bounded nondeterminism, read-once certified nondeterministic bits and tail-nondeterminism are preserved only in the component reachable from the initial configuration. 
	The remaining part of the BP may not have these properties but is not relevant to the program.
\end{proof}

\end{document}